\documentclass[11pt, onecolumn]{IEEEtran}

\usepackage{setspace}
\usepackage{amsmath, amsfonts}
\usepackage{subfigure,url}

\usepackage{xcolor}
\usepackage{xkeyval}
\usepackage{tikz}

\newtheorem{theorem}{Theorem}
\newtheorem{lemma}[theorem]{Lemma}
\newtheorem{definition}[theorem]{Definition}

\newtheorem{proposition}[theorem]{Proposition}

\newcommand{\argmax}{\operatornamewithlimits{argmax}}
\newcommand{\argmin}{\operatornamewithlimits{argmin}}

\begin{document}
\onehalfspacing

\tikzstyle voja grid=[style=help lines, color=gray, ultra thin, step=0.5cm]

\title{On Competing Wireless Service Providers}
\author{Vojislav Gaji\'{c}, Jianwei Huang, and Bixio Rimoldi
\thanks{V.~Gaji\'{c} and B.~Rimoldi are with Mobile Communications Laboratory, Ecole Polytechnique F\'ed\'erale de Lausanne, Lausanne, Switzerland, email: \{vojislav.gajic, bixio.rimoldi\}@epfl.ch. J.~Huang is with Department of Information Engineering, The Chinese University of Hong Kong, Shatin, Hong Kong, email:jwhuang@ie.cuhk.edu.hk.}}

\maketitle

\begin{abstract}

We consider a situation where wireless service providers compete for heterogenous wireless users. The users differ in their willingness to pay as well as in their individual channel gains. We prove existence and uniqueness of the Nash equilibrium for the competition of two service providers, for a generic channel model. Interestingly, the competition of two providers leads to a globally optimal outcome. We extend some of the results to the case where more than two providers are competing. Finally, we provide numerical examples that illustrate the effects of various parameters on the Nash equilibrium.
\end{abstract}

\section{Introduction and Related Work}
Due to the deregulation of telecommunication industry, one can imagine that in the future wireless users will not be contractually tied to a single service provider, but be free to switch in real time to the provider offering the best tradeoff of parameters. In this work, we consider a situation where wireless service providers want to earn profit by selling limited amount of wireless resources (e.g. bandwidth, downlink power) to a group of users. The users are rational economic agents who experience different channel conditions to the base stations of different providers and differ in their willingness to pay. The providers compete for the users by adjusting their resource prices. A user may join a provider with an inferior channel if the price of the resource is low enough.

The focus of our study is to understand the user-provider association and resource allocation in a general heterogenous network. The analysis is non-trivial in general. Consider a simplified case where there are $I$ users and $J$ base stations (belonging to different providers) distributed on a 2D plane, and the quality of a channel from a user to a base station's antenna only depends on the distance. For any prices proposed by the providers, we are interested in identifying regions of the plane with a property that users from the same region prefer to connect to the same provider. These regions of preference will differ based on the assumed utility and communication models, but in general they are non-convex (e.g., see Fig.~\ref{fig:quantity21}).  The shapes can be very irregular if we assume a more realistic wireless channel model including shadowing and fading. Furthermore, {finding the equilibrium state of the system in a naive way involves a search over $O(J^{I})$ possible choices of user-provider associations}. 

The key contributions of our work include:
\begin{itemize}

\item \emph{Network Model}: We propose a tractable network model that captures the heterogeneity of wireless users under a generic channel model. {Furthermore, we consider finitely many atomic users who have finite (non infinitesimal) demand so a single user's impact on the network equilibrium can not be ignored.}

\item \emph{Analysis of Nash Equilibrium}: For the duopoly case (competition of two providers), we prove existence and uniqueness of the Nash Equilibrium, {where at most one user will purchase resource from two providers simultaneously.} Moreover, the equilibrium maximizes the total network utility under a class of strictly increasing and concave utility functions. For a general oligopoly case (competition of more than two providers), we also obtain a partial and important characterization of the equilibrium state.

\item\emph{Reduction of Model Complexity}: {We introduce a metric transformation that turns the typical non-convex regions of preferences into convex regions in a different domain. Such convex characterization allows us to find the (integer) Nash equilibrium with only polynomial complexity in the number of users}.  
\end{itemize}

The past ten years have seen an ever increasing number of research that focuses on the application of game theory and pricing in analyzing network service providers, initiated by the work of Kelly \cite{Kelly:1997}. The majority of research in the wireless setting focuses on resource allocation within the same service provider (see for example \cite{Saraydar:2002}, \cite{Marbach:2002fk}, \cite{Chiang:2004}) and on the interaction between the users of one provider (\cite{Adlakha:2007}, \cite{Etkin:2005uq}, \cite{Huang:2006}). Perhaps surprisingly, only a few works focus on the competition between providers in a wireless setting. We can distinguish two types of competitions between wireless service providers: competing on behalf of the users (such as \cite{Zhou:2005} and  \cite{Grokop:2008}) and price competition to attract users (such as \cite{Zemlianov:2005}, \cite{Sengupta:2007} and \cite{Jia:2008}\footnote{Another related work is \cite{Felegyhazi:2007a} which considers two cellular providers competing for users by changing the strength of their pilot signals.}). The latter is the subject of our work.

Until recently, the heterogeneity of the users was largely ignored. The first work that explicitly takes into account the channel differences for different users is \cite{Inaltekin:2007}. To our knowledge, our work is the first one to consider the pricing competition of providers for users who are heterogenous in both willingness to pay and the channel quality for arbitrary channel coefficients. {Also, most previous work considered nonatomic users, i.e., each user's influence on the network is small and negligible \cite{Roughgarden:2005il}. This assumption may not be realistic in practice, which motivates us to study the atomic user case that involves resource splitting among networks (see Section \ref{sec:duopoly} for details).}

{After introducing the utility functions and communication models in Section \ref{sec:model}, we discuss the single provider case in Section \ref{sec:monopoly} and duopoly case in Section \ref{sec:duopoly}.  }We summarize our progress for the oligopoly case (more than two providers) at the end of Section \ref{sec:duopoly}. We present some numerical results in Section \ref{sec:numerics} and conclude in Section \ref{sec:conclusion}.

\section{Problem Formulation and Model}
\label{sec:model}

We consider a network with a set $\mathcal{J}=\{1,\ldots,J\}$ of service providers and a set $\mathcal{I}=\{1,\ldots,I\}$ of users. {Provider $j\in\mathcal{J}$ competes with other providers by selling a fixed amount $Q_{j}$ of a perfectly divisible resource to the users in set $\mathcal{I}$ at a unit price of $p_{j}$ with an objective of revenue maximization.  User $i \in \mathcal{I}$ experiences a channel gain $h_{ij}$ to the base station of provider $j$, drawn from some continuous distribution.} A user $i$ is free to purchase resource from a provider that offers him the highest value of utility $u_{ij}$ for all $j\in\mathcal{J}$ (c.f. \cite{Sengupta:2007}), defined by
\begin{align}
u_{ij}(p_{j},q_{ij})=a_{i}\log \left( 1+\frac{q_{ij}}{g_{i}(h_{ij})}\right)-p_{j}q_{ij},
\label{eqn:utility}
\end{align}
where $a_{i}>0$ is the willingness to pay factor of user $i$, $q_{ij}$ is the amount of resource a user is purchasing from provider $j$, and $g_{i}(h)$ is the \emph{channel quality offset} function. The function $g_{i}(h)$ is decreasing and continuous in $h$; it accounts for the effect that buying the same amount of resource from different providers will have different effects on the quality of service actually obtained by the user due to the differences in the wireless channel quality.

Our choice of utility functions is not as limited as it may seem\footnote{Most of our results still hold if we let the utility function to be a strictly concave and increasing function of $q_{ij}$ (e.g. existence, uniqueness, and social optimality of the Nash equilibrium, ordering of the users); here we choose logarithmic function since we can express quantities of interest in closed form. We feel that the insights we gain by using this simplification justify it.}. Depending on the definition of resource $q_{ij}$ and the choice of function $g_{i}(h)$, the utility function $u_{ij}$ can have different physical meanings. {A more detailed discussion on the specific utility function we chose is given in Appendix \ref{app:util-function}.} 
\begin{itemize}
\item Example 1: Consider a situation where the providers are selling downlink power. Denote the set of users purchasing resource from provider $j$ as set $\mathcal{I}_{j}$. Provider $j$ can allocate resource subject to $\sum_{i\in\mathcal{I}_{j}}q_{ij}\leq Q_{j}$. By setting $g_{i}(h_{ij})=\frac{\sigma^{2}_{i}}{|h_{ij}|^{2}}$, the utility becomes $u_{ij}=a_{i}\log (1+\frac{q_{ij}|h_{ij}|^{2}}{\sigma^{2}_{i}})-p_{j}q_{ij}=a_{i}C_{ij}(q_{ij},h_{ij},\sigma_{i}^{2})-p_{j}q_{ij}$, where $C_{ij}(\cdot)$ is the channel capacity with gain $h_{ij}$ and noise variance $\sigma_{i}^{2}$ for a user with power constraint $q_{ij}$.
\item Example 2: Consider users buying the percentage of time $q_{ij}$ they are allowed to transmit exclusively on a channel, $\sum_{i\in\mathcal{I}_{j}}q_{ij}=1(=Q_{j})$, $j \in \mathcal{J}$. Assume that each user has a maximum power constraint $P_{i}$. Then, by setting  $g_{i}(h_{ij})=\frac{1}{\log (1+\frac{P_{i} |h_{ij}|^{2}}{\sigma^{2}})}=\frac{1}{C_{ij}}$, user's utility becomes $u_{ij}=a_{i} \log (1+q_{ij} C_{ij})-p_{j}q_{ij}$ (i.e. a user's utility is an increasing function of obtained rate, with diminishing returns). Similarly, our model can be used if providers are selling exclusive access to other types of resource (e.g., bandwidth, OFDM tones) and the maximum power of a user is fixed.
\end{itemize}

Here we assume users are price-takers and do not consider the impact of their choices on the providers' prices. Taking users' choices into consideration, a provider wants to maximize its revenue by optimizing its price. A very high price will drive the users to its competitor(s), and a very low price will lead to low revenue even if its resource is fully utilized. In the next section we begin our analysis by considering a single provider. This will allow us to get insight into the effects of supply and demand for a fixed set of users. In subsequent sections we will discuss how users change their associations due to provider competition.

\section{Monopoly Case}
\label{sec:monopoly}

Since we only consider one service provider here, the subscript $j$ will be dropped in this section. For a user $i$, we begin by finding the value of the resource $q_{i}^{\ast}(p)$ that maximizes $u_{i}(p, q_{i})$ as a function of the price $p$, and call this the \emph{demand} function of user $i$. Given the concavity of the utility function $u_{i}(p, q_{i})$ in $q_{i}$, it is enough to examine the first order condition $\frac{\partial u_{i}}{\partial q_{i}}=0$ and the boundary constraint $q_{i}\geq 0$, which lead to
\begin{align*}
q^{*}_{i}(p)=\argmax_{q_{i}\geq 0} u_{i}(p,q_{i})=\left(\frac{a_{i}}{p}-g_{i}(h_{i}) \right)^{+},
\end{align*}
where $(x)^{+}=\max(x,0)$.
Notice that user $i$ will have zero demand if the price is larger than $\frac{a_{i}}{g_{i}(h_{i})}$.

The total demand that the provider faces is $Q^{*}(p)=\sum_{i \in \mathcal{I}}q^{*}_{i}(p)$, where $\mathcal{I}$ is the set of all users. We define $\mathcal{I}^{+}(p)=\{i \in \mathcal{I}: \frac{a_{i}}{p}-g_{i}(h_{i}) >0\}$, to be the set of all users with strictly positive demand at price $p$. The total demand then can be rewritten as $Q^{*}(p)=\sum_{i\in\mathcal{I}^{+}(p)}q_{i}^{*}(p)=\sum_{i\in\mathcal{I}^{+}(p)}\left(\frac{a_{i}}{p}-g_{i}(h_{i}) \right)$. Since the provider cannot sell more resource than it has available, the total profit is
\begin{align*}
\Pi=\min \left(pQ,pQ^{*}(p)\right)=\min (pQ,\sum_{i\in\mathcal{I}^{+}(p)}\left( {a_{i}}-p g_{i}(h_{i}) \right) ),
\end{align*}

\begin{proposition} \label{prop1}The revenue $\Pi$ is maximized when $pQ=pQ^{*}(p)$, i.e., when the demand equals the supply. The corresponding price is
\begin{align}
p^{*}(\mathcal{I})=\frac{\sum_{i\in\mathcal{I}^{+}(p^{*})}a_{i}}{\sum_{i\in\mathcal{I}^{+}(p^{*})}g_{i}(h_{i})+Q}.
\label{eqn:opt-price}
\end{align}
\end{proposition}
{
The profit as a function of price is illustrated by the thick line in Figure \ref{fig:profit}.
\begin{figure}[h]%
\centering
\parbox{2in}{%
\begin{tikzpicture}[>=stealth, scale=1]
	\draw[voja grid] (-0.1,-0.1) grid (5.9,3.9);
	\draw (0,0) -- (5.85,3.9) (5.7,3.9) node[left] {$pQ$};
	\draw (0,3.8) -- (2,3) -- (3.5,2) -- (4.6,1) -- (5.5,0) (0,4) node[right] {$pQ^{*}$};
	\draw[->] (0,-.1) -- (0,4);
	\draw[->] (-.1,0) -- (6,0) node [right] {$p$};
	\draw (3.255,-0.05) -- (3.255,0.05) (3.255,0)  node[below] {$p^{*}(\mathcal{I})$};
	\draw[ultra thick] (0,0) -- (3.255,2.17)-- (3.5,2) -- (4.6,1) -- (5.5,0) (3.255,2.17)  node[above] {$\Pi^{*}$};
	\draw[ultra thick]  (1.8,1.5) node {$\Pi$};
\end{tikzpicture} 
\caption{Profit Maximization}%
\label{fig:profit}}
\end{figure}
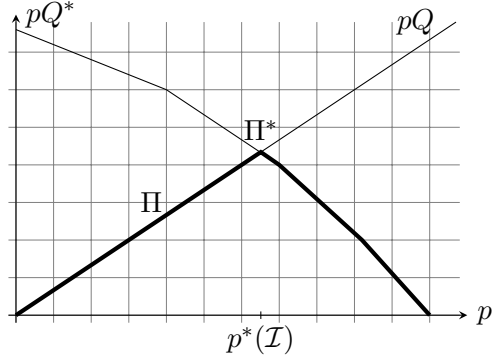

The term $pQ$ is an increasing function of the price. On the other hand, $pQ^{*}$ is decreasing in $p$ (each term in the summation is decreasing, and the set $I(p)$ is getting smaller as well). Hence, the optimal price is that for which the two terms are equal:
\begin{align}
p^{*}(\mathcal{I})=\frac{\sum_{\mathcal{I}^{+}(p^{*})}a_{i}}{\sum_{\mathcal{I}^{+}(p^{*})}g_{i}(h_{i})+Q}.
\label{eqn:opt-price}
\end{align}
}
Keep in mind that the price $p^{\ast}$ depends on the provider's total available resource $Q$. Proposition \ref{prop1} shows that if the provider charges a price higher than $p^{*}$, not all of the resource will be sold, which results in lower than optimal profits. On the other hand, a provider that charges a price $p<p^{*}$ will face a demand higher than the supply.  We call $p^{*}(\mathcal{I})$ the \emph{optimal} price in the monopoly case. The main characteristic of optimal price is that no resource is left unsold. The price $p^{*}$ is unique and we can compute it in $O(I)$ steps.  {The algorithm is given in Appendix \ref{app:opt-price}.} This is because the set $\mathcal{I}^{+}(p^{*})$ is known only once $p^{*}$ is determined, so $p^{*}$ in general cannot be computed in one calculation. The following lemma gives us a key property of the optimal price.

\begin{lemma}
For any two sets of users $\mathcal{I}$ and $\mathcal{I'}$ such that $\mathcal{I}\subset \mathcal{I'}$, $p^{*}(\mathcal{I'})\geq p^{*}(\mathcal{I})$.
\label{lem:monotone}
\end{lemma}
\begin{proof}
{The proof of Lemma \ref{lem:monotone} is given in Appendix \ref{app:monotonicity}.} 
\end{proof}

{Lemma \ref{lem:monotone} corresponds to the economic intuition that more demand leads to a higher price.}
\section{Duopoly Case}
\label{sec:duopoly}
In this section we consider two competing providers. The solution concepts and analysis techniques can be generalized to the case of more than two providers (i.e., Oligopoly). When facing competition, a provider cannot assume that all users will buy resource from it and simply charge the optimal price to maximize profit. For any two prices $p_{1}$ and $p_{2}$ announced by the providers, the users will be divided into users that prefer provider 1 ($\mathcal{I}_{1}(p_{1},p_{2})$ or simply $\mathcal{I}_{1}$) and users that prefer provider 2 ($\mathcal{I}_{2}$). Providers can now try to charge the optimal prices according to (\ref{eqn:opt-price}) for their respective set of users, giving rise to new prices $p^{*}_{1}(\mathcal{I}_{1})$ and $p^{*}(\mathcal{I}_{2})$, but these new prices may give an incentive to some users to change their provider affiliation, which will lead to new prices and so on. Our objective is to characterize the stable state in which providers have no incentive to change the price, and users have no incentive to switch between providers.

We begin by explaining how users decide which provider to join.
We assume that, for a given proposed price, the users are able to obtain the amount of resource that maximizes their utility, i.e. $q^{*}_{ij}(p_{j})=\left(\frac{a_{i}}{p_{j}}-g_{i}(h_{ij})\right)^{+}$. \footnote{This assumption makes sense if we focus on equilibrium analysis, which is the case in this paper. It may be violated when we consider the dynamics to reach the equilibrium, which is our future work. }We are interested in finding the provider that gives the user a greater maximum utility: $j^{*}_{i}=\argmax_{j\in\{1,2\}} u_{ij}(p_{j},q^{*}_{ij}(p_{j}))$. The following result gives a simple criterion for finding $j^{*}_{i}$.
\begin{lemma}
For given prices $p_{1}$ and $p_{2}$, user $i$ will join provider $j^{*}_{i}=\argmin_{j\in \{1,2\}}p_{j}g_{i}(h_{ij})$. In the case of equality, w.l.o.g.\ we assume that the user joins provider 1.
\label{lem:users-pref}
\end{lemma}
{\begin{proof}
The proof of Lemma \ref{lem:users-pref} is given in Appendix \ref{app:easyUserCondition}.
\end{proof}}

\noindent
It is clear that the selection of provider can be made by a simple ordering of $p_{j}g_{i}(h_{ij})$ and is uniquely determined for fixed prices. We define the set $\mathcal{I}_{j}(p_{1},p_{2})$ as the set of users that prefer provider $j$ to the competing provider $\bar{j}$:
\begin{align*}
\mathcal{I}_{j}(p_{1},p_{2})=\{i\in \mathcal{I}: p_{j}g_{i}(h_{ij}) <  p_{\bar{j}}g_{i}(h_{i\bar{j}})\}.
\end{align*}
There is an alternative and more informative way of expressing $\mathcal{I}_{j}$. Let us define $\alpha_{i} ={g_{i}(h_{i1})}/{g_{i}(h_{i2})}$, and w.l.o.g.\ assume that users are ordered such that $\alpha_{i}<\alpha_{i+1}$ for all $i \in \{1,\ldots,I-1\}$\footnote{Since $g(\cdot)$ depends on the channel condition which is drawn from a continuous distribution, then $Pr(\alpha_{i}= \alpha_{j})=0$ for some $i\neq j$.}. With some abuse of notation, we now write simply $\mathcal{I}_{j}(\frac{p_{2}}{p_{1}})$ instead of $\mathcal{I}_{j}(p_{1},p_{2})$ and have the following equivalent characterization of the sets:
\begin{align}
\mathcal{I}_{1}\left(\frac{p_{2}}{p_{1}}\right)=&\{i\in \mathcal{I}: \alpha_{i} \leq \frac{p_{2}}{p_{1}}  \}\; \text{ and } \; \mathcal{I}_{2}\left(\frac{p_{2}}{p_{1}}\right)=\{i\in \mathcal{I}: \alpha_{i} > \frac{p_{2}}{p_{1}}  \}. \label{eqn:cut}
\end{align}
Hence, user $i$ can be characterized by the point $\alpha_{i}$ on the real line, and the price ratio $\nu=\frac{p_{2}}{p_{1}}$ represents a cut, as illustrated in Figure  \ref{fig:partition}. Users on the left of the cut (users $1$ to $k$) choose provider 1, and users on the right of the cut (users $k+1$ to $I$) choose provider 2.

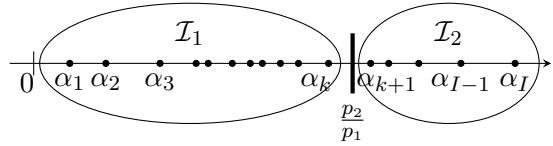
\begin{figure}[h]%
\centering
\parbox{3.6in}{%
\begin{tikzpicture}[>=stealth, scale=0.8]
\foreach \x in {0, 0.6,1.5,2.1,2.3,2.7,3,3.2,3.5,3.8,4.3,5,5.3,5.8,6.5,7.4}
	\draw[fill=black] (\x,0) circle  (0.05cm);
	\draw (0,0) node[below]{$\alpha_{1}$};
	\draw (0.6,0) node[below]{$\alpha_{2}$};
	\draw (1.5,0) node[below]{$\alpha_{3}$};
	\draw (4.1,0) node[below]{$\alpha_{k}$};
	\draw (5.27,0) node[below]{$\alpha_{k+1}$};
	\draw (6.5,0) node[below]{$\alpha_{I-1}$};
	\draw (7.4,0) node[below]{$\alpha_{I}$};	
	\draw (2,0) ellipse (2.5cm and 1cm);
	\draw (6.3,0) ellipse (1.5cm and 1cm);
	\draw[<-] (8,0) -- (-1,0) node[below right]{$0$};
	\draw(-0.6,-0.2) -- (-0.6,0.2);
	\draw (2,0.5) node{$\mathcal{I}_{1}$};
	\draw (6.3,0.5) node{$\mathcal{I}_{2}$};
\draw[ultra thick] (4.7,0.5) -- (4.7,-0.5) node[below]{$\frac{p_{2}}{p_{1}}$};
\end{tikzpicture}
\caption{Partitioning of users according to their preferences}%
\label{fig:partition}}%
\end{figure}

There are two implications of Lemma \ref{lem:users-pref} and Figure \ref{fig:partition}.
\begin{itemize}
\item \emph{Dimension reduction}: in the duopoly case, each user is characterized by two parameters:  $g_{i}(h_{i1})$ and $g_{i}(h_{i2})$. Lemma \ref{lem:users-pref} allows us to summarize them into a single parameter $\alpha_{i}$, and obtain the user division based on the ordering of this single parameter. This can be generalized to the oligopoly case where we can reduce the dimensionality from $J$ to $J-1$.
\item \emph{Region convexification}: the regions of preferences shown in Fig.~\ref{fig:partition} are always convex, whereas the physical regions of user-provider association in the 2D plane can be non-convex (e.g., see Fig.~\ref{fig:quantity21}). Convexification facilitates the analysis, and this can also be generalized to the oligopoly case.
\end{itemize}

We are now ready to define the Wireless service provider Competition Game (WCG), which is as a one-shot simultaneous move game with complete information. The two providers are the players. Each player $j$ chooses its price $p_{j}$ from the strategy set $\mathcal{P}_{j}=[0,\infty)$, and achieves a utility calculated as the total revenue $p_{j}\max\left(Q_{j},\sum_{i\in\mathcal{I}_{j}(p_{2}/p_{1})} q_{ij}^{\ast}(p_{j})\right)$.

We are interested in finding out if there exists an equilibrium situation, defined as follows:
\begin{definition}
A Nash equilibrium (NE) for the WCG is a pair of prices $(p^{NE}_{1},p^{NE}_{2})$ such that neither provider can make more profit by unilaterally changing its price.
\end{definition}

Next we describe an algorithm that can find the Nash equilibrium if it exists. We vary the price ratio parameter $\nu=p_{2}/p_{1}$ continuously from $0$ to $\alpha_I$. When $\nu = 0$ all users join provider $2$; when $\nu = \alpha_{I}$ all users join provider $1$.  As $\nu$ varies from $0$ to $\alpha_I$, we generate all possible splits of the user population into subsets $\mathcal{I}_1 (\nu)$ (users that prefer provider $1$) and $\mathcal{I}_2 (\nu)$ (users that prefer provider $2$). Here $\nu$ should be seen as a parameter to generate the partitions that are of potential interest to us and not as an actual price ratio. We define the \emph{optimal price ratio} as
\begin{equation}\label{eq:ratio}
 \mu(\nu)=\frac{p^*_{2}(\mathcal{I}_2 (\nu))}{p^* _{1}(\mathcal{I}_1 (\nu))},
\end{equation}
where the optimal prices $p_{j}^{\ast}s$ are calculated according to (\ref{eqn:opt-price}). If the user partitions generated by $\nu$ are the same as those generated by the optimal price ratio $ \mu(\nu)$ induced by $\nu$, then the optimal price ratio generates stable user partitions.  Figures \ref{fig:MCE} and \ref{fig:noMCE} plot a sample function $\mu(\nu)$ and $f(\nu)=\nu$.  Notice that $\mu(\nu)$ is a piece-wise constant and non-increasing function, and it changes value at $\nu = \alpha_i$, $i = 1, \ldots, I$.

\begin{figure}[h]%
\centering
\parbox{2.5in}{%
\begin{tikzpicture}[>=stealth, scale=0.7]

\filldraw[fill=black] (0.5,0) circle (1pt)  ++(0,0) node [below, scale=0.7] {$\alpha_{1}$};
\filldraw[fill=black] (3,0) circle (1pt) ++(-0.15,0) node [below, scale=0.7] {$\alpha_{k}$};
\filldraw[fill=black] (3.5,0) circle (1pt) ++(0,0) node [below, scale=0.7] {$\alpha_{k+1}$};
\filldraw[fill=black] (5.5,0) circle (1pt) ++(0,-0.02) node [below, scale=0.7] {$\alpha_{I}$};
\draw[-] (0,0) -- (6,3.8) node[left] {$f(\nu)=\nu$};

\filldraw[fill= black] (1,3.8) -- (0.5,3.8) circle (1pt)  ;
\filldraw[fill= black] (1.3,3.5) -- (1,3.5) circle (1pt)  ;
\filldraw[fill= black] (3,2.3) -- (2.6,2.3) circle (1pt)  ;
\filldraw[fill= black] (3.5,2) -- (3,2) circle (1pt)  ;
\filldraw[fill= black] (3.8,1.7) -- (3.5,1.7) circle (1pt)  ;
\filldraw[fill= black] (5.5,0.5) -- (5,0.5) circle (1pt)  ;
\filldraw[fill= black] (5.5,0) circle (1pt) ;

\filldraw[fill= black] (-0.05,3.8)-- ++(0.1,0) node [left, scale=0.7] {$\mu(\alpha_{1})$} ;
\filldraw[fill= black] (-0.05,2)-- ++(0.1,0)  node [left, scale=0.7] {$\mu(\alpha_{k})$};
\filldraw[fill= black] (-0.05,0)-- ++(0.1,0)  node [left, scale=0.7] {$\mu(\alpha_{I})$};

\filldraw[fill= black] (1.5,3.1) circle (0.5pt);
\filldraw[fill= black] (1.8,2.9) circle (0.5pt);
\filldraw[fill= black] (2.1,2.7) circle (0.5pt);

\filldraw[fill= white] (4.2,1.2) circle (0.5pt);
\filldraw[fill= white] (4.4,1.05) circle (0.5pt);
\filldraw[fill= white] (4.6,0.9) circle (0.5pt);

\draw[->] (-0.2,0) -- (6.2,0) node[right] {$\nu$};
\draw[->] (0,-0.2) -- (0,4.2) node[above] {$\mu(\nu)$};
\end{tikzpicture}
\caption{Integer Market Clearing Equilibrium}%
\label{fig:MCE}}%
\qquad \qquad
\begin{minipage}{2.5in}%
\begin{tikzpicture}[>=stealth, scale=0.7]

\filldraw[fill=black] (0.5,0) circle (1pt)  ++(0,0) node [below, scale=0.7] {$\alpha_{1}$};
\filldraw[fill=black] (3,0) circle (1pt) ++(-0.15,0) node [below, scale=0.7] {$\alpha_{l-1}$};
\filldraw[fill=black] (3.5,0) circle (1pt) ++(0,0) node [below, scale=0.7] {$\alpha_{l}$};
\filldraw[fill=black] (5.5,0) circle (1pt) ++(0,-0.02) node [below, scale=0.7] {$\alpha_{I}$};
\draw[-] (0,0) -- (6,3.8) node[left] {$f(\nu)=\nu$};

\filldraw[fill= black] (1,3.8) -- (0.5,3.8) circle (1pt)  ;
\filldraw[fill= black] (1.3,3.5) -- (1,3.5) circle (1pt)  ;
\filldraw[fill= black] (3,2.5) -- (2.6,2.5) circle (1pt)  ;
\filldraw[fill= black] (3.5,2.3) -- (3,2.3) circle (1pt)  ;
\filldraw[fill= black] (3.8,2) -- (3.5,2) circle (1pt)  ;
\filldraw[fill= black] (5.5,0.5) -- (5,0.5) circle (1pt)  ;
\filldraw[fill= black] (5.5,0) circle (1pt) ;

\filldraw[fill= black] (-0.05,3.8)-- ++(0.1,0) node [left, scale=0.7] {$\mu(\alpha_{1})$} ;
\filldraw[fill= black] (-0.05,2.3)-- ++(0.1,0)  node [left, scale=0.7] {$\mu(\alpha_{l-1})$};
\filldraw[fill= black] (-0.05,2)-- ++(0.1,0) node [left, scale=0.7] {$\mu(\alpha_{l})$};
\filldraw[fill= black] (-0.05,0)-- ++(0.1,0)  node [left, scale=0.7] {$\mu(\alpha_{I})$};

\filldraw[fill= black] (1.5,3.1) circle (0.5pt);
\filldraw[fill= black] (1.8,2.9) circle (0.5pt);
\filldraw[fill= black] (2.1,2.7) circle (0.5pt);

\filldraw[fill= white] (4.2,1.2) circle (0.5pt);
\filldraw[fill= white] (4.4,1.05) circle (0.5pt);
\filldraw[fill= white] (4.6,0.9) circle (0.5pt);

\draw[->] (-0.2,0) -- (6.2,0) node[right] {$\nu$};
\draw[->] (0,-0.2) -- (0,4.2) node[above] {$\mu(\nu)$};
\end{tikzpicture}
\caption{Fractional Market Clearing Equilibrium}%
\label{fig:noMCE}%
\end{minipage}%
\end{figure}%

Figure \ref{fig:MCE} shows the case when $\mu(\cdot)$ has a fixed point (there exists $\nu^{*}$ such that $\mu(\nu^{*})=\nu^{*}$). This is the only point that is stable, and from Figures \ref{fig:MCE} and \ref{fig:noMCE} it is clear that there can be at most one such point. Hence, we can distinguish between two cases: (i) Type 1 as in  Fig.~\ref{fig:MCE}: there is a unique fixed point for $\mu(\cdot)$, and (ii) Type 2 as in Fig.~\ref{fig:noMCE}: there is no fixed point for $\mu(\cdot)$.

We first consider Type 1.

\begin{proposition}{\bf{[Integer Market Clearing Equilibrium]}}
{If there exist $\nu^{*}$ and $k$ such that  $\nu^{*}\in [\alpha_{k},\alpha_{k+1})$ and $\mu(\nu^{*})=\nu^{*}$,} then the prices $(p^{*}_{1}(\{1,\ldots, k \}),p^{*}_{2}(\{k+1,\ldots,I\}))$ form a Nash Equilibrium. We call this equilibrium the integer \emph{market clearing equilibrium} (MCE) and the corresponding prices \emph{equilibrium prices}.
\label{pr:integerMCE}
\end{proposition}

{\begin{proof}
The condition $\mu(\nu^{*})=\nu^{*}$, where $\nu^{*}\in[\alpha_{k},\alpha_{k+1})$ is equivalent to having $\alpha_{1}\cdots <\alpha_{k}<\frac{p^{*}_{2}(\{k+1,\ldots,I\})}{p^{*}_{1}(\{1,\ldots,k\})}< \alpha_{k+1}<\ldots<\alpha_{I}$ (see Figure~\ref{fig:MCE}). Then users $1$ through $k$ prefer provider 1, and users $k+1$ through $I$ prefer provider 2. Also, the provider 1 is charging the optimal price $p_{1}= p^{*}_{1}(\{1,\ldots,k\})$, so the demand the provider is facing is $Q_{1}$, i.e. exactly equal to the supply. Hence, the price is profit maximizing and the provider has no incentive to deviate. Similar reasoning is true for provider 2. Finally, since the providers are charging optimal prices, each user is getting the utility maximizing amount of resource from his preferred provider.
\end{proof}}
The market clearing equilibrium has some nice properties: the supply equals the demand for each provider, the providers maximize their profits, and the users are getting their utility functions maximized by their preferred provider. {Since $\mu(\nu)$ is a piece-wise constant function, we can show that the fixed point $\nu^{\ast}$ can be found in at most $I+1$ steps, where each step involves examining one of the following $I+1$ intervals: $[0,\alpha_{1}),[\alpha_{1},\alpha_{2}),\ldots,[\alpha_{I-1},\alpha_{I}),$ and $[\alpha_{I},\infty)$. This means that the search algorithm only has a complexity polynomial in the number of users.}

The Type 2 case, pictured in Figure \ref{fig:noMCE}, implies that there exists an index $ l$ for which  $\mu(\alpha_{l-1})>\alpha_{l}>\mu(\alpha_{l})$, i.e.:
\begin{align}
\frac{p^{*}_{2}(\{ l, l+1, \ldots,I\})}{p^{*}_{1}(\{1,2,\ldots, l-1\})} \stackrel{(a)}{>} \alpha_{ l} \stackrel{(b)}{>} \frac{p^{*}_{2}(\{ l+1, \ldots,I\})}{p^{*}_{1}(\{1,2,\ldots, l\})}.
\label{eqn:teeter}
\end{align}
The inequality $(a)$ in equation (\ref{eqn:teeter}) implies that user $ l$ prefers provider 1 if the optimal price ratio is computed assuming he is associated with provider 2, while inequality $(b)$ implies that he prefers provider 2 if he is associated with provider 1.  Hence, user $l$ has an incentive to switch away from a provider as soon as he joins it. In this case, we call user $l$ the \emph{undecided} user.

The undecided user problem arises when the influence of the undecided user on the equilibrium price is non-negligible\footnote{\noindent As mentioned in the introduction, in our model the users do not consider their own impact on the system when choosing the best provider; users are price-takers \cite{Kelly:1997}, and not price anticipators \cite{Johari:2004}.}. For example, in any situation where $I=1$ (there is only one user), being the user of provider 1 renders the other provider more desirable since $p^{*}_{2}(\o)=0$.
An undecided user can exists even for a large number of users, although the impact of a single user on the equilibrium price is smaller in that case.

Mathematically, the undecided user issue arises due to the discontinuity of the optimal price ratio in (\ref{eq:ratio}). This is because each user can only purchase resource from one provider.

If we relax this assumption and allow user $l$ to get resource from both providers, then he can split his demand in such a way that $\frac{p_{2}}{p_{1}}=\alpha_{l}$ and user $l$ is indifferent as to which provider he purchases resource from $\left(p_{1}g_{l}(h_{l1})=p_{2}g_{l}(h_{l2})\right)$. The following theorem asserts that there is a unique way in which the undecided user splits his demand to ensure market stability.
\begin{proposition}{\bf{[Fractional Market Clearing Equilibrium]}} Assume that there exists a user $l$ such that $\mu(\alpha_{l-1})>\alpha_{l}>\mu(\alpha_{l})$ and this user can purchase the resource from both providers. In particular, user $l$ purchases $\epsilon q^{*}_{l1}(p_{1})$ resource from provider 1 and $(1-\epsilon) q^{*}_{l2}(p_{2})$ from provider 2. Then there exist unique $p^{*}_{1}$, $p^{*}_{2}$ and $\epsilon^{*}$ such that $\frac{p^{*}_{1}}{p^{*}_{2}}=\alpha_{l}$, the total demand equals the total supply for each provider, and each user obtains the resource maximizing utility from its preferred provider.
\label{pro:eps-equilib}
\end{proposition}
{\begin{proof}
The proof of Proposition \ref{pro:eps-equilib} is given in Appendix \ref{app:eps-equilib}.
\end{proof}
}

Notice that the fractional equilibrium requires somewhat more than the definition of the Nash equilibrium: in addition to charging the correct prices, the providers must split the load of the undecided user in the correct ratio $\epsilon^{*}$. In practice this may be difficult to realize, but the fractional equilibrium is still significant in that it provides a generalization of the Nash equilibrium to the Type 2 case. The condition that a user may purchase from both providers may seem to be at odds with the exposition so far. However, if we assume from the beginning that the utility of a user is $u_{i} \stackrel{\triangle}{=}a_{i}\log \left( 1+\sum_{j\in\mathcal{J}}\frac{q_{ij}}{g_{i}(h_{ij})}\right)-\sum_{j\in\mathcal{J}}p_{j}q_{ij}$ (instead of defining a utility function towards each provider as in (\ref{eqn:utility})), it can be shown that both Type 1 and Type 2 can be incorporated in the same framework and there is always a unique Nash equilibrium. In particular, all users except the undecided user will choose to purchase resource from a single provider even though they have the freedom to purchase from both providers. 

We can now define the total network utility $U_{T}$ as the summation of the users' utilities and the providers' profits, i.e. $U_{T}=\sum_{i \in \mathcal{I}}^{}u_{i}+\sum_{j=1}^{2}\Pi_{j}=\sum_{i \in \mathcal{I}} a_{i}\log \left( 1+\sum_{j\in\mathcal{J}}\frac{q_{ij}}{g_{i}(h_{ij})}\right)$, and state the main result for the duopoly case:

\begin{theorem}
\label{th:PoA}
Assume that users can purchase the resource from both providers. A WCG has a unique outcome (either integer or fractional MCE) with the property that providers have no incentive to change the price and users obtain their maximum utilities. Furthermore, the total network utility $U_{T}$ is maximized at the equilibrium.
\end{theorem}
{\begin{proof}
The proof of Theorem \ref{th:PoA} is given in Appendix \ref{app:socialOptimality}.
\end{proof}
}

Theorem \ref{th:PoA} can be proved by comparing the results of Propositions \ref{pr:integerMCE} and \ref{pro:eps-equilib} with the optimality (KKT) condition of the total network utility maximization problem. This shows that the competition among service providers does not affect the social welfare\footnote{In other words, the Price of Anarchy (\cite{RoTa02}) is zero.}. We can also show that this result holds for any strictly increasing, strictly concave utility function of the users.

For the oligopoly case ($J\geq3$) we can find the integer MCE (when it exists) using an efficient centralized algorithm (polynomial in the number of users), demonstrate its uniqueness, and represent users in $J-1$ dimensional space such that the regions of preference are convex. Our ongoing work focuses on finding fractional MCE and the development of decentralized algorithms for finding the MCE.

\section{Numerical Results}
\label{sec:numerics}

\begin{figure}[ht] 
   \centering
   \parbox{3.2in}{
   \includegraphics[width=2.6in]{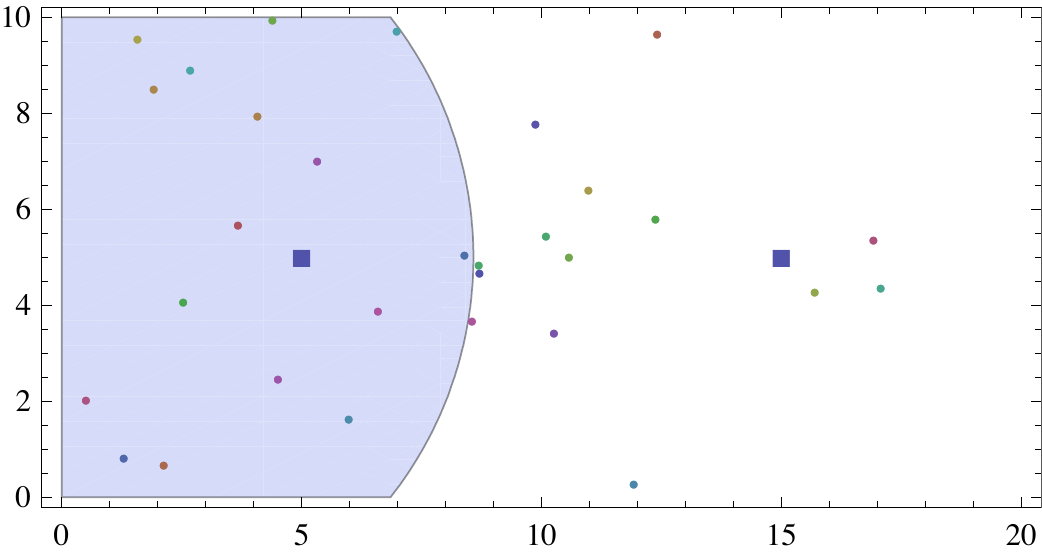}
   \caption{Providers have different supplies: $Q_{1}=\frac{1}{10}Q_{2}$.}
   \label{fig:quantity21}}
   \qquad \qquad
   \parbox{3in}{
    \includegraphics[width=2.6in]{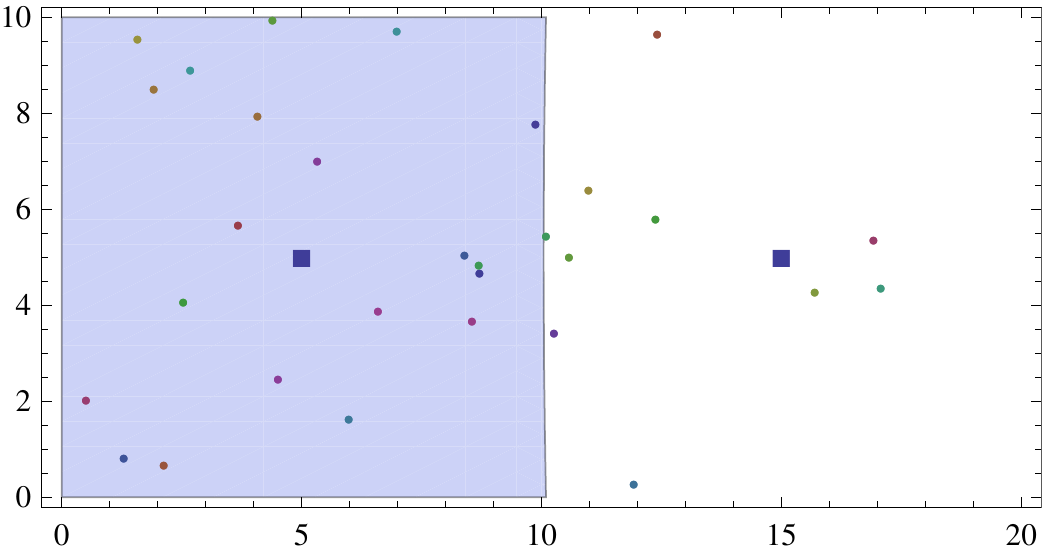}
   \caption{Providers have equal supplies: $Q_{1}=Q_{2}$.}
   \label{fig:equalQuant}}
\end{figure}
In this section we present some numerical results showing the effects of different parameters in the WCG. We consider two providers selling downlink transmit power ($Q_{1},Q_{2}$) to 30 users contained in a $10 \times 20$ area. The utility function for user $i$ is $u_{ij}=a_{i} \log (1+\frac{q_{ij}}{d^{\beta}_{ij}})-p_{j}q_{ij}$, where $d_{ij}$ is the distance to base station $j$ and $\beta$ is the path loss exponent (this corresponds to $g_{i}(h_{ij})=\frac{1}{h_{ij}}=\frac{1}{d^{\beta}_{ij}}$). The baseline value we take is  $\beta=3$. Users are placed uniformly in the area and their willingness to pay coefficients $a_{i}$s are chosen from a uniform distribution.

Figures \ref{fig:quantity21} and \ref{fig:equalQuant} show the effect of supply on the equilibrium association. In both figures an integer MCE is obtained. The network parameters are the same in both cases, except $Q_{1}$ in Fig.~\ref{fig:equalQuant} is 10 times larger than in Fig.~\ref{fig:quantity21}. Base stations of different providers are represented by squares and users by dots. The whole plane is divided into two disjoint areas: the shaded area indicates the region with users who are associated to provider 1 (on the left) and the remaining unshaded area indicates the users with provider 2 (which is nonconvex). It is clear that by increasing supply $Q_{1}$, provider 1 attracts more users in Fig.~\ref{fig:equalQuant} compared with Fig.~\ref{fig:quantity21}.

\begin{figure}[ht] 
   \centering
   \parbox{3.2in}{
   \includegraphics[width=2.8in]{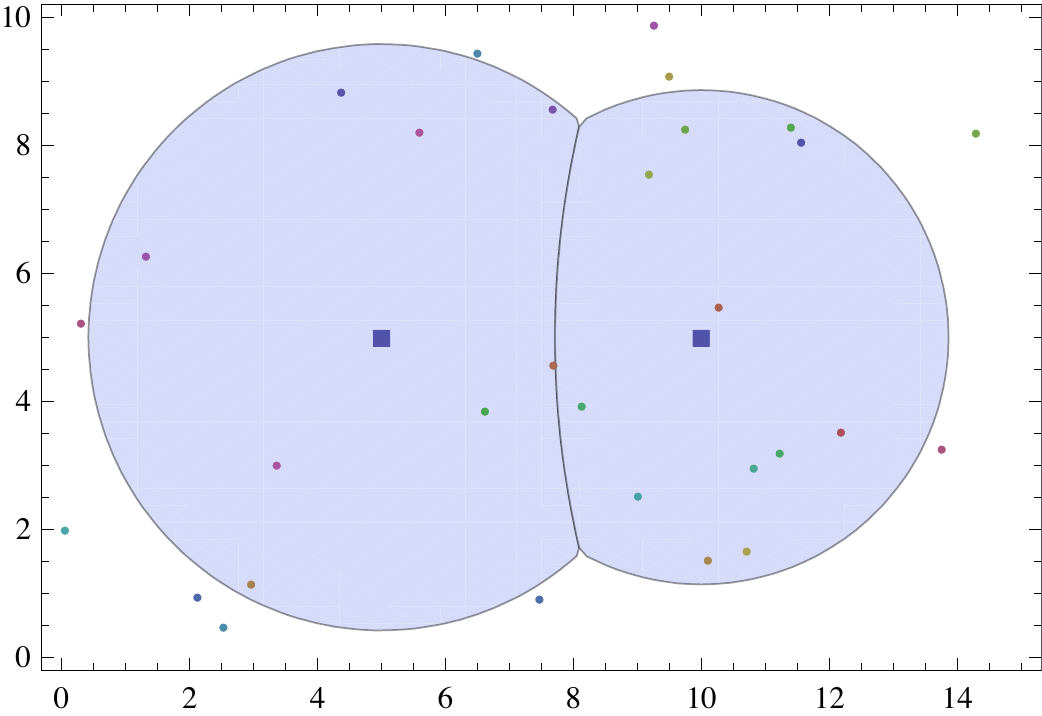}
   \caption{Regions with positive demands.}
   \label{fig:nonzerodemand}}
   \qquad \qquad
   \parbox{3in}{
    \includegraphics[width=2.6in]{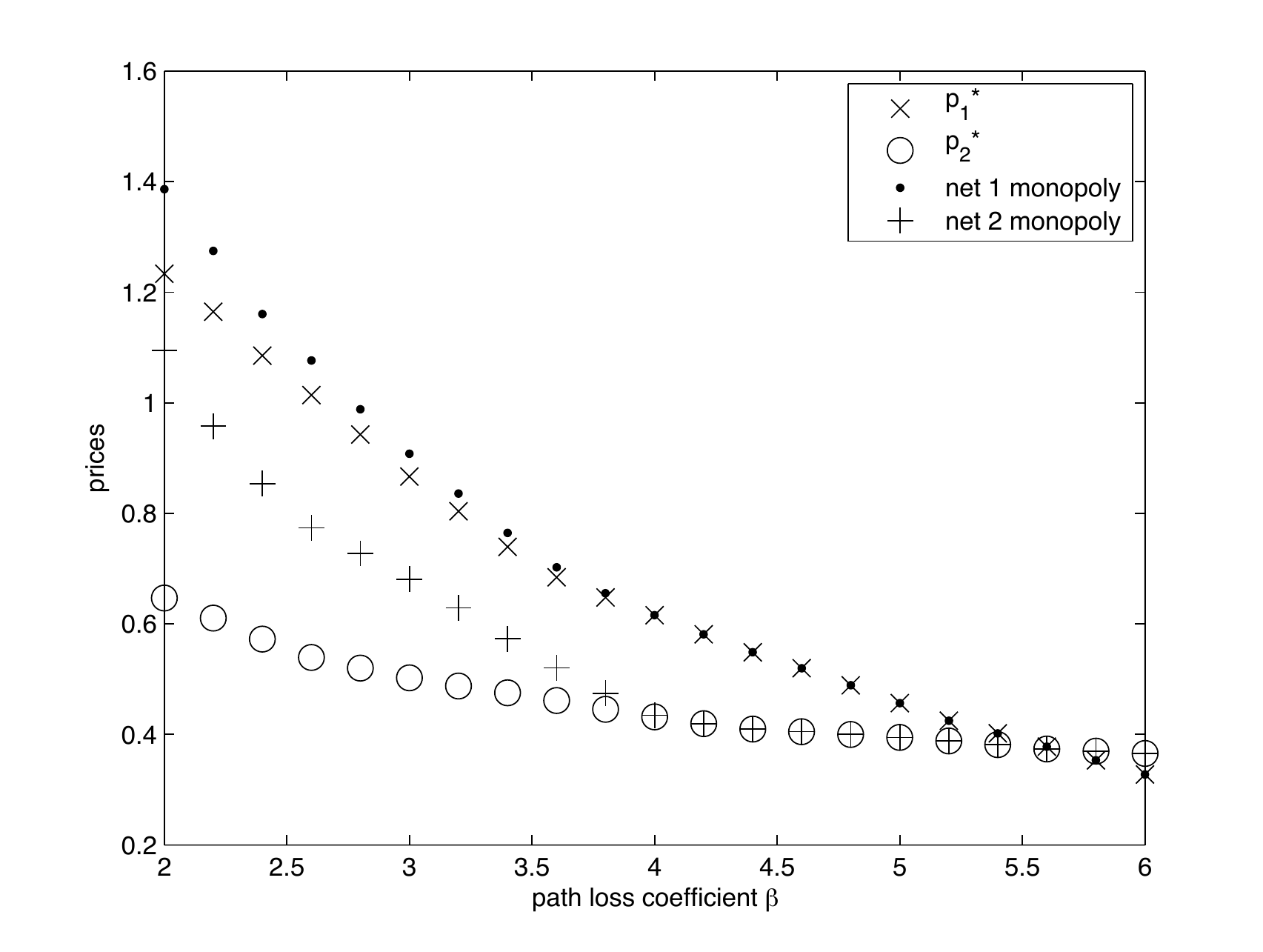}
   \caption{Prices as a function of the path loss coefficent}
   \label{fig:pathloss}}
\end{figure}

In Fig.~\ref{fig:nonzerodemand} the whole plane is divided into three regions: the unshaded region where the users have zero demand from both providers, and the two shaded regions where users have positive demand from one of the providers. Here both providers have the same amount of supply ($Q_{1}=Q_{2}$), while the users have equal willingness to pay. Since there are less users around the base station of provider 1 (on the left), the equilibrium price $p_{1}^{\ast}$ is smaller than $p_{2}^{\ast}$ and the preference region (shaded area) for provider 1 is larger compared with that of provider 2 (since lower price can attract users who are further away).

In Fig.~\ref{fig:pathloss} we plot the equilibrium prices ($p_{1}^{\ast}$ and $p_{2}^{\ast}$) as the path loss exponent varies from $\beta$ from $2$ to $6$. We also plot the monopolistic optimal prices for providers 1 and 2. As one would expect, the duopoly prices are lower than the monopoly ones. For large values of $\beta$ the duopoly situation looks very much like two monopolies, since signal attenuation is so strong that users have positive demand from only the closest provider and thus there is no competition among providers.

\section{Conclusion and Future Directions}
\label{sec:conclusion}

In this paper, we study the competition of wireless service providers in a heterogeneous user environment. We characterize the equilibrium state of the competition game through proper metric transformation and dimension reduction of the prameters, and obtain a surprising result that competition leads to a globally optimal outcome under some commonly used users' utility functions and a generic channel model. Our future work will focus on extending the results we have obtained for the duopoly case to the competition of an arbitrary number of providers, developing decentralized algorithms that result in equilibrium state, and analyzing the competitions under more general utility functions.

\begin{appendix}
\subsection{About the users' utility function}
\label{app:util-function} 
We present in more detail the merits of using the specific form of the utility function given in (\ref{eqn:utility}). We begin by presenting a specific example. Consider two base stations that are in the business of selling downlink power to wireless users. Base station $j\in\{1,2\}$ has $P_{j}$ power to sell to the paying users. The cost per unit of power $p_{j}$ is set by the base station $j$. Users $i$ experiences a channel $h_{ij}$ to base station $j$, so that the maximum rate that a user can obtain is $R^{*}_{ij}(P_{ij})=\log (1+\frac{P_{ij}|h_{ij}|^{2}}{\sigma^{2}_{i}})$, where $P_{ij}$ is the power obtained from network $j$, and $\sum_{i=1}^{I}P_{ij}\leq P_{j}$. At the same time, the user has to pay $p_{j}P_{ij}$ to the network. Hence, a user faces a tradeoff in purchasing the downlink power from the network that offers a better channel and getting it from a network that charges less. 
We assume that the users are economic agents whose goal is to maximize their utility (quality of service/happiness). One way to model the tradeoff that the users are experiencing is by assuming that the utility function of a user $i$ is:
\begin{align*}
u_{ij}=a_{i}R^{*}_{ij}-p_{j}P_{ij}=a_{i}\log (1+\frac{P_{ij}|h_{ij}|^{2}}{\sigma^{2}_{i}})-p_{j}P_{ij},
\end{align*}
where $a_{i}>0$ is a user-specific parameter indicating how happy a user is with the wireless resources. For example, for users $k$ and $l$ with $a_{k}>a_{l}$ (but equally good channels), user $k$ will demand more power. 

The user problem is then to choose the network $j^{*}_{i}$ that offers greater utility:
\begin{align*}
{j}^{*}_{i}=\argmax_{j\in\{1,2\}} u_{ij},
\end{align*}
which is a function of the prices proposed by the two networks, and the quantity of the power purchased. The other parameters in the problem are fixed: $a_{i}$'s by the users and $h_{ij}$'s by nature.

In particular, we assume that user $i$ will choose the provider $j$ that offers him a higher value of $u_{ij}$ (c.f. \cite{Sengupta:2007}), where:
\begin{align*}
u_{ij}=a_{i}\log \left( 1+\frac{q_{ij}}{g_{i}(h_{ij})}\right)-p_{j}q_{ij}.
\end{align*}
In our earlier example the resource $Q$ is power, and $g_{i}(h_{ij})=\frac{\sigma_{i}^{2}}{|h_{ij}|^{2}}$. The \emph{channel quality offset} function $g_{i}(h_{ij})$ that we introduce allows us to differentiate the effect that buying a resource from a particular provider has on the quality of service actually obtained by the user. The only assumption on the channel quality offset $g_{i}(h)$ is that it is a decreasing function of $h$. In our example, the factor $g_{i}(\cdot)$  captures the fact that buying a unit of power in general does not provide the user with the same amount of rate for networks that offer him different channel conditions.

Aside from the channel quality offset parameter, our model is similar to other models existing in the literature. The satisfaction of communication users is often modeled as an increasing, concave function of the desirable resource, from which the cost of obtaining the resource is subtracted. The concavity of the utility function reflects the fact that the initial resource is more important to a user (the more a user has, the less he needs). The utility function we use is one of such kind. Keeping the utility as a general concave function would have been more appropriate, but we chose logarithm for its convenience and since its use permits us to write many quantities of interest in compact form.

\subsection{Optimum price computation}
\label{app:opt-price}
In this section we show how the optimal price $p^{*}(I)$ can be computed for any given set of users $I$. Recall that 
\begin{align}
p^{*}(\mathcal{I})=\frac{\sum_{\mathcal{I}^{+}(p^{*})}a_{i}}{\sum_{\mathcal{I}^{+}(p^{*})}g_{i}(h_{i})+Q}.
\label{eqn:opt-price}
\end{align}
As shown in Section \ref{sec:monopoly}, the optimal price is the solution of the equation $pQ=\sum_{i \in \mathcal{I}}\left( a_{i}-p g_{i}(h_{i})\right)^{+} $. Consider the price function 
\begin{align}
p^{f}(\mathcal{I})=\frac{\sum_{\mathcal{I}}a_{i}}{\sum_{\mathcal{I}}g_{i}(h_{i})+Q},
\label{eqn:fictitious-price}
\end{align}
which is the solution to the equation $pQ=\sum_{i \in \mathcal{I}}\left( a_{i}-p g_{i}(h_{i})\right) $, i.e. it is the optimal price if users are allowed to purchase negative amount of resource (a user purchasing negative resource would increase the supply of the seller). We will call $p^{f}(\mathcal{I})$ the \emph{fictitious} price function.  Unlike the optimal price function, the fictitious price function is easily calculated from (\ref{eqn:fictitious-price}). Furthermore, we can use the fictitious price function to calculate the optimal price function. Before defining the algorithm rigorously, we give a high level explanation. The idea is to find a subset of users for which the fictitious price is equal to the optimal price.

The fictitious price equals the optimal price if and only if $\frac{a_{i}}{p^{f}(\mathcal{I})}-g_{i}(h_{i})\geq 0$ for all $i\in \mathcal{I}$. In general, this is not the case since some users will have negative demand. Then, we proceed by excluding those ``bad'' users from the set $\mathcal{I}$ and by recalculating the fictitious price for the users with positive valuation only (the ``good'' users). Removing bad users increases the fictitious price so the set of good users can only get smaller. We terminate the procedure when we find the price for which there are no more bad users. In each step we remove at least one user, and there has to be at least one good user remaining. Hence this algorithm terminates in at most $I-1$ steps with the set of good users who have positive demand at the optimal price. The formal proof follows.

 We begin by calculating $p^{f}(\mathcal{I})$. Then, we check whether all of the users in $\mathcal{I}$ have positive demand (i.e. if $\frac{a_{i}}{p^{f}(\mathcal{I})}-g_{i}(h_{i})\geq 0$ for all $i\in \mathcal{I}$). If yes, then $\frac{a_{i}}{p^{f}(\mathcal{I})}-g_{i}(h_{i})=(\frac{a_{i}}{p^{f}(\mathcal{I})}-g_{i}(h_{i}))^{+}$ for all $i \in \mathcal{I}$ and the fictitious price is equal to the optimal price. If not, there exists at least one user $k$  who has negative demand ($\frac{a_{k}}{p^{f}(\mathcal{I})}-g_{k}(h_{k})<0$). We compare $p^{f}(\mathcal{I})$ and $p^{f}(\mathcal{I} \setminus \{k\} )$ (to ease notation, we denote by $\mathcal{I}'=\mathcal{I} \setminus \{k\}$ the set of users without user $k$):
 \begin{align}
p^{f}(\mathcal{I}')-p^{f}(\mathcal{I})= & \frac{\sum_{i\in \mathcal{I}'} a_{i}}{\sum_{i\in \mathcal{I}'}g_{i}(h_{i})+Q} - \frac{\sum_{i\in \mathcal{I}} a_{i}}{\sum_{i\in \mathcal{I}}g_{i}(h_{i})+Q} \nonumber \\
= & \frac{p^{f}(\mathcal{I}')\left(g_{k}(h_{k})-\frac{a_{k}}{p^{f}(\mathcal{I}')}\right)}{\sum_{i\in I}g_{i}(h_{i})+Q} > 0.
\label{eqn:removeUser}
\end{align}
Hence, removing a user with negative demand increases the fictitious price. Notice that, since $p^{f}(\mathcal{I}')>p^{f}(\mathcal{I})$, then $\frac{a_{k}}{p^{f}(\mathcal{I}')}-g_{k}(h_{k})<\frac{a_{k}}{p^{f}(\mathcal{I})}-g_{k}(h_{k})<0$, i.e. a user with negative demand at price $p^{f}(\mathcal{I})$ will still have negative demand at price $p^{f}(\mathcal{I}')$. The same statement holds if we remove several bad users at the same time (since we can always think of this as removing bad users one by one). At the same time, if the price is changed to $p^{f}(\mathcal{I}')$, all good users will have their demand decreased, and some may see their demand go negative. Hence, additional users may need to be removed and price recaculated. These steps are repeated until the price is such that there are no users with negative demand.

This procedure is guaranteed to terminate with at least one user. To see this, note that a single user always has positive demand when he is the network's only customer. 

Now we can define the optimal price computation algorithm. We initiate $\mathcal{I}^{*}=\mathcal{I}$ at the beginning of the algorithm.

\begin{tt}
Algorithm 1:\\
Begin
\begin{itemize}
\item[1.]  Calculate  $p^{f}(\mathcal{I}^{*})$.
\item[2.] Find $\hat{\mathcal{I}}=\{i \in \mathcal{I}^{*}: \frac{a_{i}}{p^{f}(I)}-g_{i}(h_{i})<0\}$
\item[3.] If $|\hat{\mathcal{I}}|>0$ set $\mathcal{I}^{*}:=\mathcal{I}^{*} \setminus \hat{\mathcal{I}}$ and go to step 1.
\item[4.] $p^{*}(\mathcal{I})=p^{f}(\mathcal{I}^{*})$.
\end{itemize}
End\\
\end{tt}
Notice that, at the end of the algorithm, 
\begin{align*}
\sum_{i \in \mathcal{I}}\left(\frac{a_{i}}{p^{*}(\mathcal{I})}-g_{i}(h_{i})\right)^{+}= &\sum_{i \in \mathcal{I}^{*}}\left(\frac{a_{i}}{p^{*}(\mathcal{I})}-g_{i}(h_{i})\right)^{+}+\sum_{i \in \mathcal{I} \setminus \mathcal{I}^{*}}\left(\frac{a_{i}}{p^{*}(\mathcal{I})}-g_{i}(h_{i})\right)^{+}\\
= &\sum_{i \in \mathcal{I}^{*}}\left(\frac{a_{i}}{p^{f}(\mathcal{I}^{*})}-g_{i}(h_{i})\right),
\end{align*}
since $\left(\frac{a_{i}}{p^{f}(\mathcal{I}^{*})}-g_{i}(h_{i})\right)^{+}=0$ for all $i\in \mathcal{I} \setminus \mathcal{I}^{*}$. Hence, the fictitious price for $\mathcal{I}^{*}$ is the optimal price for $\mathcal{I}$.

\subsection{Monotonicity of the optimal price (proof of Lemma \ref{lem:monotone})}
\label{app:monotonicity}
In this appendix, we prove that adding a new user to the user set can only increase the optimal price. Suppose that we already computed $p^{*}(\mathcal{I})$ and we wish to compute $p^{*}(\mathcal{I}')$, where $\mathcal{I}'=\mathcal{I} \cup \{k\}$. Without loss of generality, assume that $\frac{a_{i}}{p^{*}(\mathcal{I})}-g_{i}(h_{i})\geq0$ for all $i\in \mathcal{I}$ (otherwise, we can always restrict ourselves to such a set of users). Notice that adding user $\{k\}$ does not change anything if $\frac{a_{k}}{p^{*}(\mathcal{I})}-g_{k}(h_{k})\leq0$ (then user $k$ has zero demand, and the optimal price is unchanged. This is different than the fictitious price from Appendix \ref{app:opt-price} which decreases if a bad user is added). So, we assume $\frac{a_{k}}{p^{*}(\mathcal{I})}-g_{k}(h_{k})>0$ Then, following (\ref{eqn:removeUser})

\begin{align}
p^{f}(\mathcal{I}')-p^{*}(\mathcal{I})=& \frac{\sum_{i\in \mathcal{I}'} a_{i}}{\sum_{i\in \mathcal{I}'}g_{i}(h_{i})+Q} - \frac{\sum_{i\in \mathcal{I}} a_{i}}{\sum_{i\in \mathcal{I}}g_{i}(h_{i})+Q} \nonumber \\
= & \frac{p^{*}(\mathcal{I})\left(\frac{a_{k}}{p^{*}(\mathcal{I})}-g_{k}(h_{k})\right)}{\sum_{i\in \mathcal{I}'}g_{i}(h_{i})+Q} > 0.
\label{eqn:addUser}
\end{align}
In principle, in order to find $p^{*}(\mathcal{I}')$ we need to run Algorithm 1. Here, it suffices to notice that in each step of Algorithm 1 we remove only bad users, which can only increase the fictitious price. Hence, $p^{*}(\mathcal{I})\geq p^{f}(\mathcal{I}')$ and therefore $p^{*}(\mathcal{I}')>p^{*}(\mathcal{I})$ if $\frac{a_{k}}{p^{*}(\mathcal{I})}-g_{k}(h_{k})>0$. In general, for $\frac{a_{k}}{p^{*}(\mathcal{I})}-g_{k}(h_{k})\in \mathbb{R}$ we have $p^{*}(\mathcal{I}')\geq p^{*}(\mathcal{I})$. 

Hence, the optimal price function is non-decreasing if we add a single user. Since any group of users can be added by adding users one by one, the result of Lemma \ref{lem:monotone} holds.

\subsection{Proof of Lemma \ref{lem:users-pref}}
\label{app:easyUserCondition}

Assume that $q^{*}_{i1},q^{*}_{i2}>0$, i.e. $a_{i}>p_{1}g_{i}(h_{i1})$, $a_{i}>p_{1}g_{i}(h_{i1})$. We investigate the difference between $u_{i1}(p_{1},q^{*}_{i1})$ and $u_{i2}(p_{2},q^{*}_{i2})$ for any fixed $(p_{1},p_{2})$. By substituting $q^{*}_{ij}$ into $u_{ij}(p_{j},q^{*}_{ij})$ we get:
\begin{align*}
u_{i1}-u_{i2} = & a_{i}\log \left( \frac{a_{i}}{p_{1}g_{i}(h_{i1})}\right)-a_{i}+p_{1}g_{i}(h_{i1}) - a_{i}\log \left( \frac{a_{i}}{p_{2}g_{i}(h_{i2})}\right)+a_{i}-p_{2}g_{i}(h_{i2})\\
=& - a_{i}\log \left( \frac{p_{1}g_{i}(h_{i1})}{p_{2}g_{i}(h_{i2})}\right)+p_{2}g_{i}(h_{i2})\left( \frac{p_{1}g_{i}(h_{i1})}{p_{2}g_{i}(h_{i2})}-1\right).
\end{align*}
When $p_{1}g_{i}(h_{i1})=p_{2}g_{i}(h_{i2})$, this expression is equal to 0, in which case the user is indifferent between the two networks. Assume $p_{1}g_{i}(h_{i1})<p_{2}g_{i}(h_{i2})$ and perform the change of variables $1-x= \frac{p_{1}g_{i}(h_{i1})}{p_{2}g_{i}(h_{i2})}$, $x\in(0,1)$. Then,
\begin{align*}
u_{i1}-u_{i2} = & -a_{i}\log \left( 1-x \right) +p_{2}g_{i}(h_{i2}) \left( 1-x -1\right)\\
	\stackrel{(a)}{>} & a_{i}\,x - p_{2}g_{i}(h_{i2})\,x = p_{2} \left(\frac{a_{i}}{p_{2}}-g_{i}(h_{i2})\right)x\\
	= &  p_{2} q^{*}_{i2}\,x \\
	> &0,
\end{align*}
where $(a)$ follows from $\log (1-x)<-x,\; x\in (0,1)$ . Hence, $u_{i1}>u_{i2}$ when $p_{1}g_{i}(h_{i1})<p_{2}g_{i}(h_{i2})$. Similarly, from symmetry it follows that $u_{i1}<u_{i2}$ when $p_{1}g_{i}(h_{i1})>p_{2}g_{i}(h_{i2})$. This concludes the proof of Lemma \ref{lem:users-pref} for a user that has positive demand towards both networks.

It remains to prove the cases when the user has zero demand towards one or both of the networks. When $q^{*}_{i1}=q^{*}_{i2}=0$, the user does not want to communicate, so by convention he can prefer network 1. Assume that, $p_{1}g_{i}(h_{i1})<p_{2}g_{i}(h_{i2})$, but that $q^{*}_{i1}=0$ and $q^{*}_{i2}>0$. Then user $i$ would prefer network 2 (since it gives him positive utility, and network 1 gives him zero utility). But this case can never arise if $p_{1}g_{i}(h_{i1})<p_{2}g_{i}(h_{i2})$. To see this, notice that $q^{*}_{i1}=0$ implies $p_{1}q^{*}_{i1}=0$ which in turn implies $a_{i}-p_{1}g_{i}(h_{i1})\leq 0$, whereas $q^{*}_{i2}>0$ implies  $a_{i}-p_{2}g_{i}(h_{i2})> 0$. Hence we have
\begin{align*}
a_{i}-p_{1}g_{i}(h_{i1})\leq 0 &< a_{i}-p_{2}g_{i}(h_{i2}) \text{ or} \\
p_{1}g_{i}(h_{i1}) > & p_{2}g_{i}(h_{i2}),
\end{align*}
which is a contradiction. Finally, if $p_{1}g_{i}(h_{i1})<p_{2}g_{i}(h_{i2})$ and $q_{i1}^{*}>0$ but $q_{i2}^{*}=0$, user $i$ prefers network 1 since it gives him positive utility, while network 2 gives him zero utility. Similar reasoning goes when $p_{1}g_{i}(h_{i1})>p_{2}g_{i}(h_{i2})$.

\subsection{Proof of Theorem \ref{pro:eps-equilib}}
\label{app:eps-equilib}

To prove the theorem, we need to find a way to split the resources required by user $k$ in such a way that the user $k$ is indifferent between the two networks. 

By the condition of the theorem $\mu(k-1)>\alpha_{k}>\mu(k)$, i.e. $\frac{p^{*}_{2}(\{k, \ldots,N\})}{p^{*}_{1}(\{1,\ldots,k-1\})} > \alpha_{k}> \frac{p^{*}_{2}(\{k+1, \ldots,N\})}{p^{*}_{1}(\{1,\ldots,k\})}$. Let $\mathcal{I}_{1}=\{1,\ldots, k-1\}$ and $\mathcal{I}_{2}=\{k+1,\ldots,M\}$ and assume users from $\mathcal{I}_{1}$ are associated with network 1, and users from $\mathcal{I}_{2}$ are associated with network 2. The idea is to extend the definition of optimal prices to include partial demand of the undecided user $k$ so that the function $\mu(\cdot)$ can be continuous. We will provide a definition of $p^{\epsilon}_{1}$ such that $p^{0}_{1}=p^{*}_{1}(\mathcal{I}_{1})$, $p^{1}_{1}=p_{1}^{*}(\mathcal{I}_{1}\cup \{k\})$, and changing the value of $\epsilon$ from $0$ to $1$ sweeps the interval $(p^{*}_{1}(\mathcal{I}_{1}),p^{*}_{1}(\mathcal{I}_{1}\cup \{k\}))$ (and similarly for $p^{\epsilon}_{2}$). In other words, $p^{\epsilon}_{1}(\mathcal{I}_{1})$ is a strictly increasing function of $\epsilon$, and for each value $p \in (p^{*}_{1}(\mathcal{I}_{1}),p^{*}_{1}(\mathcal{I}_{1}\cup \{k\}))$ there exists a unique $\epsilon_{p}$ such that $p^{\epsilon_{p}}_{1}(\mathcal{I}_{1})=p$. 

Notice that users whose demand is already zero at $p^{*}_{1}(\mathcal{I}_{1})$ will have zero demand at $p^{*}_{1}(\mathcal{I}_{1}\cup \{k\})$ so we can safely ignore them. Assume, without loss of generality, that all users from $\mathcal{I}_{1}$ (respectively, $\mathcal{I}_{2}$) have strictly positive demand when facing price $p^{*}_{1}(\mathcal{I}_{1})$ (respectively, $p^{*}_{2}(\mathcal{I}_{2})$). Then, for $j = 1,2$: 
\begin{align*}
p^{*}_{j}(\mathcal{I}_{j})=\frac{\sum_{i \in \mathcal{I}_{j}}a_{i}}{\sum_{i \in \mathcal{I}_{j}}g_{i}(h_{ij})+Q_{j}}.
\end{align*}
When increasing the price from $p^{*}_{1}(\mathcal{I}_{1})$ to $p^{*}_{1}(\mathcal{I}_{1}\cup \{k\})$ some users may see their demand go down to zero. The nuisance in accounting for this is that there is no closed-form expression for the optimal prices in general (the only expression at our disposal is the one for the fictitious price. Recall that the fictitious price and the optimal price are the same if and only if no user has negative demand). 

To begin with, we assume that $\frac{a_{i}}{p^{*}_{j}(\mathcal{I}_{j}\cup \{k\})}-g_{i}(h_{ij})\geq 0$ for all $i \in \mathcal{I}_{j}$, $j =1,2$ ¥ (i.e. the increase in price caused by the addition of user $k$ does not make any of the users give up communicating). This is not true in general, but considering this case first will make it easier to understand the general case where some of the users may see their demand drop to zero. Naturally, $q_{kj}^{*}>0$, $j=1,2$ otherwise user $k$ cannot change the price and there is no undecided user issue. Then, 
\begin{align*}
p^{*}_{j}(\mathcal{I}_{j} \cup \{k\})=\frac{\sum_{i \in \mathcal{I}_{j}}a_{i}+a_{k}}{\sum_{i \in \mathcal{I}_{j}}g_{i}(h_{ij})+g_{k}(h_{kj})+Q_{j}},
\end{align*}
and it makes sense to define the \emph{fractional} equilibrium price $p^{\epsilon}_{j}(\mathcal{I}_{j})$ (with some abuse of notation) as
\begin{align*}
p^{\epsilon}_{j}(\mathcal{I}_{j}) \stackrel{\triangle}{=} p^{*}_{j}(\mathcal{I}_{j} \cup \{\epsilon k\})=\frac{\sum_{i \in \mathcal{I}_{j}}a_{i}+\epsilon a_{k}}{\sum_{i \in \mathcal{I}_{j}}g_{i}(h_{ij})+\epsilon g_{k}(h_{kj})+Q_{j}},
\end{align*}
for all $\epsilon \in [0,1]$. In other words $p^{\epsilon}_{j}(\mathcal{I}_{j})$ is obtained as the optimal price if user $k$ demands quantity $\epsilon q^{*}_{kj}$ from network $j$ instead of $q^{*}_{kj}$. We can see that $p^{\epsilon}_{j}$ is continuous and strictly increasing for all $\epsilon \in [0,1]$. For each $p_{j} \in (p^{*}_{1}(\mathcal{I}_{1}), p^{*}_{1}(\mathcal{I}_{1}\cup \{k\})$ there exists an $\epsilon$ such that $p_{j}=p^{\epsilon}_{j}$. Then, we can define $\mu(s)$ for $s \in [k-1,k]$  as
\begin{align*}
\mu(s)\stackrel{\triangle}{=} \mu(k-1+\epsilon k)=\frac{p^{\epsilon}_{1}}{p^{1-\epsilon}_{2}},
\end{align*}
where $k=\lceil s \rceil$ and $\epsilon = s-k+1$. Since $p^{\epsilon}_{1}$ is continuous and strictly increasing in  $\epsilon$, and $p^{1-\epsilon}_{2}$ is continuous and strictly decreasing, then $\mu(s)$ is continuous and strictly increasing for $s \in [k-1,k]$. Hence there exists a unique $s^{*}$ and unique $\epsilon^{*}=s^{*}-(k-1)$ such that $\mu(s^{*})=\frac{p^{\epsilon^{*}}_{1}}{p^{1-\epsilon^{*}}_{2}}=\alpha_{k}=\frac{g_{{k}}(h_{k2})}{g_{{k}}(h_{k1})}$. These prices are obtained by user $k$ asking for $\epsilon^{*} q^{*}_{k1}=\epsilon^{*}(\frac{a_{i}}{p^{\epsilon^{*}}_{1}}-g_{k}(h_{k1}))$ from network 1 and $(1-\epsilon^{*}) q^{*}_{k2}=(1-\epsilon^{*})(\frac{a_{i}}{p^{1-\epsilon^{*}}_{2}}-g_{k}(h_{k2}))$ from network 2. Notice that in this case the utility of user $k$ is:
\begin{align*}
u_{k}=&\log \left(1+\frac{q_{k1}}{g_{k}(h_{k1})}+\frac{q_{k2}}{g_{k}(h_{k2})}\right) - p_{1}q_{k1}- p_{2}q_{k2} \\
=&\log \left(1+\frac{\epsilon^{* }q^{*}_{k1}}{g_{k}(h_{k1})}+\frac{(1-\epsilon^{*})q^{*}_{k2}}{g_{k}(h_{k2})}\right) - p^{\epsilon^{*}}_{1}\epsilon^{*}q^{*}_{k1}- p^{1-\epsilon^{*}}_{2}(1-\epsilon^{*})q^{*}_{k2} \\
=& \log \left( \epsilon^{*} \frac{a_{i}}{p^{\epsilon^{*}}_{1}g_{k}(h_{k1})}+(1-\epsilon^{*})\frac{a_{i}}{p^{1-\epsilon^{*}}_{2}g_{k}(h_{k2})}\right)-\epsilon^{*}p^{\epsilon^{*}}_{1}q^{*}_{k1}- (1-\epsilon^{*})p^{1-\epsilon^{*}}_{2}q^{*}_{k2} \\
= &u^{*}_{k1}=u^{*}_{k2}, 
\end{align*}
where the equalities in the last line come from the fact that $p^{\epsilon^{*}}_{1}g_{k}(h_{k1})=p^{1-\epsilon^{*}}_{2}g_{k}(h_{k2})$. Hence, user $k$ is indeed obtaining the maximum utility. It is important to notice that, at prices $p^{\epsilon^{*}}_{1}$ and $p^{1-\epsilon^{*}}_{2}$ user $k$ can split his demand in any way (i.e. choose any value of $\epsilon \in [0,1]$) and still obtain the same utility. However, only by choosing $\epsilon^{*}$ he is making sure that the networks actually have the required amount of resource he asks for, i.e. that the total demand for each network equals the total supply.

As for the remaining users, notice that for all $i \in \mathcal{I}_{1}$, it holds that $\alpha_{i} < \alpha_{k}=\mu(s^{*})$ so users associated with network 1 are indeed associated with their preferred network, and getting utility maximizing resource from it. Similar statement is true for users from $\mathcal{I}_{2}$. Hence, with prices set to $p^{\epsilon^{*}}_{1}$ and $p^{1-\epsilon^{*}}_{2}$ the networks have no incentive to change prices, and the users are associated to their preferred network (expect user $k$ who is indifferent between the two and gets resource from both).

What remains is to consider the general case where there exist users whose demand drops to zero for some $p_{j} \in (p^{*}_{j}(\mathcal{I}_{j}),p^{*}_{j}(\mathcal{I}_{j}\cup \{k\}))$. In this case we need to define the $\epsilon$ price in a piecewise manner. 

Let $\breve{\mathcal{I}}_{j}=\{i \in \mathcal{I}_{j}: \frac{a_{i}}{p^{*}_{1}(\mathcal{I}_{1}\cup \{k\})}-g_{i}(h_{ij})<0\}\,j=1,2$, i.e. $\breve{\mathcal{I}}_{j}$ is the set of all users who are dissuaded from communicating if user $k$ joins network $j$. Also, let $\breve{p}_{i}\in (p^{*}_{j}(\mathcal{I}_{j}),p^{*}_{j}(\mathcal{I}_{j}\cup \{k\}))$ be such that $\frac{a_{i}}{\breve{p}_{i}}-g_{i}(h_{ij})=0$ for $i \in \breve{\mathcal{I}}_{j},\,j=1,2$. The price $\breve{p}_{i}$ is the lowest price for which the demand of user $i$ becomes zero (it is understood which network the user $i$ belongs to).

We will focus on network 1 only, but the same argument holds for network 2. As mentioned before, the optimal price for the set of users $\mathcal{I}_{1}=\{1,\ldots, k-1\}$ is
\begin{align*}
p^{*}_{1}(\mathcal{I}_{1})=\frac{\sum_{i \in \mathcal{I}_{1}}a_{i}}{\sum_{i \in \mathcal{I}_{1}}g_{i}(h_{i1})+Q_{1}}.
\end{align*}
Assume that $i_{1},\ldots,i_{|\breve{\mathcal{I}}_{1}|}$ are such that $\breve{p}_{i_{1}}<\breve{p}_{i_{1}}<\ldots<\breve{p}_{i_{|\breve{\mathcal{I}}_{1}|}}$. Then we can define 
\begin{align*}
p^{\epsilon}_{1}(\mathcal{I}_{1}) =\frac{\sum_{i \in \mathcal{I}_{1}}a_{i}+\epsilon a_{k}}{\sum_{i \in \mathcal{I}_{1}}g_{i}(h_{i1})+\epsilon g_{k}(h_{k1})+Q_{1}}, \text{ for } \epsilon \in [0,\breve{\epsilon}_{i_{1}}], 
\end{align*}
where $\breve{\epsilon}_{i_{1}}$ is such that 
\begin{align*}
\breve{p}_{i_{1}} =\frac{\sum_{i \in \mathcal{I}_{1}}a_{i}+\breve{\epsilon}_{i_{1}} a_{k}}{\sum_{i \in \mathcal{I}_{1}}g_{i}(h_{i1})+\breve{\epsilon}_{i_{1}} g_{k}(h_{k1})+Q_{1}}, 
\end{align*}
i.e.
\begin{align*}
\breve{\epsilon}_{i_{1}}=\frac{\breve{p}_{i_{1}} \left( \sum_{i\in \mathcal{I}_{1}} g_{i}(h_{i1})+Q_{1}\right)-\sum_{i\in \mathcal{I}_{1}}a_{i}}{a_{k}-\breve{p}_{i_{1}}g_{k}(h_{k1})}.
\end{align*}
For $\epsilon \in [0,\breve{\epsilon}_{i_{1}})$ no user has his demand drop down to zero, and at $\epsilon=\breve{\epsilon}_{i_{1}}$ the price $p^{\epsilon}_{1}(\mathcal{I}_{1})$ exactly equals $\breve{p}_{i_{1}} =\frac{a_{i_{1}}}{g_{i_{1}}(h_{i_{1}1})}$. At this point the user $i_{1}$ (the user with the lowest valuation for the resource) can be excluded from the optimal price formula since
\begin{align*}
\breve{p}_{i_{1}} =\frac{\sum_{i \in \mathcal{I}_{1}}a_{i}+\breve{\epsilon}_{i_{1}} a_{k}}{\sum_{i \in \mathcal{I}_{1}}g_{i}(h_{i1})+\breve{\epsilon}_{i_{1}} g_{k}(h_{k1})+Q_{1}}=\frac{\sum_{i \in \mathcal{I}_{1}\setminus \{i_{1}\}}a_{i}+\breve{\epsilon}_{i_{1}} a_{k}}{\sum_{i \in \mathcal{I}_{1}\setminus \{i_{1}\} }g_{i}(h_{i1})+\breve{\epsilon}_{i_{1}} g_{k}(h_{k1})+Q_{1}},
 \end{align*}
which follows directly from  $\breve{p}_{i_{1}} =\frac{a_{i_{1}}}{g_{i_{1}}(h_{i_{1}1})}$. Similarly, we can now define
\begin{align*}
p^{\epsilon}_{1}(\mathcal{I}_{1}) =\frac{\sum_{i \in \mathcal{I}_{1}\setminus \{i_{1}\}}a_{i}+{\epsilon} a_{k}}{\sum_{i \in \mathcal{I}_{1}\setminus \{i_{1}\} }g_{i}(h_{i1})+{\epsilon} g_{k}(h_{k1})+Q_{1}}, \epsilon \in [\breve{\epsilon}_{i_{1}}, \breve{\epsilon}_{i_{2}}]. 
 \end{align*}
In general, we can define the fractional equilibrium price as:
\begin{align*}
p^{\epsilon}_{1}(\mathcal{I}_{1}) =\frac{\sum_{i \in \mathcal{I}_{1}\setminus \{i_{1},\ldots,i_{l}\}}a_{i}+{\epsilon} a_{k}}{\sum_{i \in \mathcal{I}_{1}\setminus \{i_{1},\ldots,i_{l}\} }g_{i}(h_{i1})+{\epsilon} g_{k}(h_{k1})+Q_{1}}, \epsilon \in [\breve{\epsilon}_{i_{l}}, \breve{\epsilon}_{i_{l+1}}], 
 \end{align*}
where $\breve{\epsilon}_{i_{0}}=0$, $\breve{\epsilon}_{i_{|\breve{\mathcal{I}}_{1}|+1}}=1$ and for $l \in \{1,\ldots,|\breve{\mathcal{I}}_{1}|\}$: 
\begin{align*}
\breve{\epsilon}_{i_{l}}=\frac{\breve{p}_{i_{l}} \left( \sum_{i\in \mathcal{I}_{1}\setminus \{i_{1},\ldots,i_{l}\}} g_{i}(h_{i1})+Q_{1}\right)-\sum_{i\in \mathcal{I}_{1}\setminus \{i_{1},\ldots,i_{l}\}}a_{i}}{a_{k}-\breve{p}_{i_{l}}g_{k}(h_{k1})}.
\end{align*}
Note that $p^{\epsilon}_{1}(\mathcal{I}_{1})$ defined in this way is continuous and strictly increasing. We can similarly define $p^{\epsilon}_{2}(\mathcal{I}_{2})$ to be continuous and strictly increasing for all $\epsilon \in [0,1]$ and then define the fractional equilibrium price ratio $\mu(s)=\frac{p^{\epsilon}_{1}(\mathcal{I}_{1})}{p^{1-\epsilon}_{2}(\mathcal{I}_{2})}$ as before. The rest of the proof is identical as for the case where no users see their demand drop down to zero.

\subsection{Social optimality of provider competition (proof of Theorem \ref{th:PoA})}
\label{app:socialOptimality}

Let $v_{i}(x)$ be a strictly concave, strictly increasing, continuous differentiable everywhere function on $\mathbb{R}$. Let $x_{i}=\frac{q_{i1}}{g_{i1}}+\frac{q_{i2}}{g_{i2}}$, where $g_{ij}\stackrel{\triangle}{=}g_{i}(h_{ij})$, i.e. the user's utility is a function of the obtained resource, scaled by the quality offset factor for the provider he is buying this resource from. If networks are charging prices $p_{1},p_{2}$ per unit resource, then the utility of user $i$ is $u_{i}(q_{i1},q_{i2})=v_{i}(\frac{q_{i1}}{g_{i1}}+\frac{q_{i2}}{g_{i2}})-p_{1}q_{i1}-p_{2}q_{i2}$. By choosing $v_{i}(x)=a_{i}\log(1+x)$ we recover the utility functions used in the paper. However, the proof is more general so we give it in those terms.\\

\begin{theorem}
\label{thm:system}
There exist unique prices $p_{1},p_{2}$ and, corresponding to these prices,  a unique demand vector $\mathbf{q}=[q_{11}\; q_{12}\; \cdots \; q_{I1}\; q_{I2}]$ such that $q_{i1},q_{i2}$ are the utility maximizing resource quantities for user $i$ at prices $p_{1},p_{2}$ and such that $\sum_{i=1}^{I}q_{i1}=Q_{1}$ and $\sum_{i=1}^{I}q_{i2}=Q_{2}$. Furthermore, there can be at most one user with $q_{i1}q_{i2}>0$. 
\end{theorem}
\begin{proof}
Consider the utility maximization problem SYSTEM$(\mathbf{v},\mathbf{g},Q)$:
\begin{align}
\max\; & \sum_{i=1}^{I}v_{i}\left(\frac{q_{i1}}{g_{i1}}+\frac{q_{i2}}{g_{i2}}\right) \label{eqn:max} \\
\text{subject to }& \sum_{i=1}^{I}q_{ij}=Q_{j},\; j=1,2 \label{eqn:clearing}\\
\text{over } & q_{ij}\geq0 \; \forall i,j \nonumber
\end{align}

Strictly speaking, the constraint in (\ref{eqn:clearing}) should be a constraint on a larger set: $\sum_{i=1}^{I}q_{ij}\leq Q_{j},\; j=1,2$. However, since the functions $v_{i}(\cdot)$ are strictly increasing, we can improve the utility sum obtained for any demand vector that satisfies $\sum_{i=1}^{I}q_{ij}< Q_{j}$ by assigning the remaining resource $\left(Q_{j}-\sum_{i=1}^{I}q_{ij}\right)$ arbitrarily to some user. Hence, the maximizing solution will be at the boundary and the constraint set is characterized by (\ref{eqn:clearing}).

The maximization problem SYSTEM is a concave maximization problem over a convex feasible region, so we know that the maximum is attained. However, there may be several maximizing values. In the following we prove that the maximizing vector $\mathbf{q}$ is unique.

The Lagrangian for this maximization problem is
\begin{align}
L(\mathbf{q},\mathbf{p})=& \sum_{i=1}^{I}v_{i}\left(\frac{q_{i1}}{g_{i1}}+\frac{q_{i2}}{g_{i2}}\right) + \sum_{j=1}^{2}p_{j}\left( Q_{j}-\sum_{i=1}^{I}q_{ij}\right).
\end{align}
When we differentiate the Lagrangian with respect to the variables of interest, we obtain:
\begin{align}
\frac{\partial v_{i}(\frac{q_{i1}}{g_{i1}}+\frac{q_{i2}}{g_{i2}})}{\partial q_{ij}}=\frac{\partial v_{i}(x_{i})}{\partial x_{i}}\frac{1}{g_{ij}}-p_{j} = &0,\;q_{ij }>0 \label{eqn:derivZero} \\
\leq & 0,\; q_{ij} =0, \label{eqn:derivLessZero}
\end{align}
where with some abuse of notation we write $\frac{\partial v_{i}(x_{i})}{\partial x_{i}}$ instead of $\frac{\partial v_{i}(x)}{\partial x}\Big\vert_{x=x_{i}=\frac{q_{i1}}{g_{i1}}+\frac{q_{i2}}{g_{i2}}}$. In particular, from (\ref{eqn:derivZero}) and (\ref{eqn:derivLessZero}) we can see that 
\begin{align*}
\frac{\partial v_{i}(x_{i})}{\partial x_{i}}\leq p_{1}g_{i1} \qquad \frac{\partial v_{i}(x_{i})}{\partial x_{i}}\leq p_{2}g_{i2} , \forall \; i \in \mathcal{I}
\end{align*}
and in particular
\begin{align*}
\frac{\partial v_{i}(x_{i})}{\partial x_{i}}= & \min(p_{1}g_{i1},p_{2}g_{i2}),\; q_{i1}>0\text{ or }q_{i2}>0 \\
\leq & \min(p_{1}g_{i1},p_{2}g_{i2}),\; q_{i1}=q_{i2}=0.
\end{align*}
Furthermore, if $p_{1}g_{i1} \neq p_{2}g_{i2}$ we have $q_{i1}q_{i2}=0$, i.e. user $i$ gets the resource from at most one of the providers. The condition $p_{1}g_{i1} = p_{2}g_{i2}$ can only be fulfilled for one user, since the probability that $\frac{g_{i1}}{g_{i2}}=\frac{g_{k1}}{g_{h2}}$ is of zero measure, summed over all pairs $i,k \in \mathcal{I}$. So, at most one user can have $q_{i1}q_{i2}>0$ for any $p_{1},p_{2}$. Also, notice that $q_{i1}=0$ if $\frac{g_{i1}}{g_{i2}}>\frac{p_{2}}{p_{1}}$ and similarly $q_{i2}=0$ if $\frac{g_{i1}}{g_{i2}}<\frac{p_{2}}{p_{1}}$. We say that the Lagrange multiplier ratio determines the preferred network of a user.

We know that a maximizing solution of SYSTEM exists, let us denote it $\mathbf{q}^{*}$. We want to prove that this solution is unique. To do this, we first show that a maximizing solution identifies a unique par of Lagrange multipliers $p^{*}_{1},p^{*}_{2}$. We then assume that there are two maximizing solution and show that this leads to a contradiction. The contradiction comes from the fact that values of the Lagrangian multipliers other than the maximizing ones will lead to values of $\mathbf{q}$ that violate the equality in (\ref{eqn:clearing}). 

Suppose that $\mathbf{q}^{*}$ is the maximizing solution of SYSTEM. We want to show that there is exactly one pair of Lagrange multipliers $\mathbf{p}^{*}$ associated with $\mathbf{q}^{*}$ and that for any other maximizing solution $\mathbf{q}'$ the associated Lagrange multipliers $\mathbf{p}'$ are different than $\mathbf{p}^{*}$. There exist at least one $i$ and $k$ such that $q^{*}_{i1}>0$ and $q^{*}_{k2}>0$ (otherwise (\ref{eqn:clearing}) is not satisfied). Assume also that $i$ and $k$ are such that $q^{*}_{i2}=0$ and $q^{*}_{k1}=0$ \footnote{Such $i$ and $k$ cannot be found only in two special cases (there is only one user with positive demand, buying all resource from both networks; or there are only two users with positive demand, one of which is buying all of the resource from one of the networks). It can be shown that a maximizing $\mathbf{q}$ yields unique $p_{1}^{*}$ and $p_{2}^{*}$ for these cases as well.}.
Then $x^{*}_{i}=\frac{q^{*}_{i1}}{g_{i1}}$ and $x^{*}_{k}=\frac{q^{*}_{k2}}{g_{k2}}$ so we can find $p^{*}_{1}$ and $p^{*}_{2}$ from (\ref{eqn:derivZero}):
\begin{align*}
p^{*}_{1}=g_{i1}\frac{\partial v_{i}(x)}{\partial x} \Big \vert_{x=x_{i}^{*}} \text{ and } p^{*}_{2}=g_{k2}\frac{\partial v_{k}(x)}{\partial x} \Big \vert_{x=x_{k}^{*}}.
\end{align*}
The derivative of $v_{i}$ is a strictly decreasing function, so $\frac{\partial v_{i}(x)}{\partial x}$ is a bijection on $(0,\infty)$. Hence, a solution of SYSTEM corresponds to exactly one price vector $\mathbf{p}=[p^{*}_{1}\; p^{*}_{2}]$. Now consider $\mathbf{q}'\neq\mathbf{q}^{*}$. There exists at least one purchased quantity that is different in the two vectors, i.e. $q_{ij}'\neq q_{ij}^{*} \neq 0$ for some $i$. If $q_{i\bar{j}}'=0=q_{i\bar{j}}^{*}$ (user $i$ only purchases from one network), where $\bar{j}=\{1,2\}\setminus j$, then by recalling that $\frac{\partial v_{i}(x)}{\partial x}$ is a bijection we conclude that $p_{j}^{*} \neq p_{j}'$. If user $i$ purchases from both networks, then $\sum_{k \neq i}^{¥}q_{ij}'=Q_{j}-q_{ij}' \neq Q_{j}-q_{ij}^{*}=\sum_{k \neq i}^{¥}q_{ij}^{*}$. Since  $\sum_{k \neq i}^{¥}q_{ij}'\neq \sum_{k \neq i}^{¥}q_{ij}^{*}$ then there must exist another user $l$ such that $q_{lj}'\neq q_{lj}^{*}$. User $l$ only buys from one network (since there is only one user that buys from both networks, and it is user $i$) so by the previous argument we see that $p_{j}'\neq p_{j}^{*}$. We conclude that different maximizing vectors $\mathbf{q}$ correspond to different Lagrange multipliers $\mathbf{p}$.

Next we prove that only one maximizing solution exists. Suppose again that $\mathbf{q}^{*}$ is the maximizing solution of SYSTEM. 
Remember that the resource provider that the user is getting the resource from is determined solely on the basis of the Lagrange multiplier ratio $\frac{p^{*}_{1}}{p^{*}_{2}}$. We define $\mathcal{I}_{j}=\{i: q^{*}_{ij}>0\}$ where at most one user is in both sets. 

Suppose there is another maximizing solution $\mathbf{q'}$, which is obtained for Lagrange multipliers $p_{1}',p_{2}'$. There are three cases we need to consider: $\frac{p^{*}_{1}}{p^{*}_{2}}=\frac{p_{1}'}{p_{2}'}$, $\frac{p^{*}_{1}}{p^{*}_{2}}>\frac{p_{1}'}{p_{2}'}$, and $\frac{p^{*}_{1}}{p^{*}_{2}}<\frac{p_{1}'}{p_{2}'}$ (the third case boils down to the second case).

Case 1: $\frac{p^{*}_{1}}{p^{*}_{2}}=\frac{p_{1}'}{p_{2}'}$. Assume, without loss of generality, that $p^{*}_{j}<p_{j}'$, $j=1,2$. We can define the sets $\mathcal{I}'_{j}=\{i:q'_{ij}>0\}$, $j=1,2$. Since the price ratio stayed equal, then no user with positive $q_{ij}$ changed their set, i.e. $\mathcal{I}_{j}' \subset \mathcal{I}_{j}$ (some users that previously had $q_{ij}>0$ may now have $q_{ij}=0$ so they are no longer in the set). Also, no user went from $q^{*}_{ij}=0$ to $q'_{ij}>0$, which will be clear from the discussion bellow. For all (except maybe one user who has $q^{*}_{k1}>0, q^{*}_{k2}>0$) $i \in \mathcal{I}_{1}$ we have $q^{*}_{i1}>0,q^{*}_{i2}=0$ for $p^{*}_{1},p^{*}_{2}$ and $q'_{i2}  \geq 0,q_{i2}'= 0$ for $p_{1}',p_{2}'$. From (\ref{eqn:derivZero}) we know that
\begin{align*}
\frac{\partial v_{i}(x)}{\partial x}\Big\vert_{x=\frac{q^{*}_{i1}}{g_{i1}}}=p^{*}_{1}g_{i1} \text{ and } \frac{\partial v_{i}(x)}{\partial x}\Big\vert_{x=\frac{q'_{i1}}{g_{i1}}}=p_{1}'g_{i1}.
\end{align*}
Since $v_{i}$ is a strictly concave function of $x$, then its derivative is a strictly decreasing function of $x$. Hence, $p^{*}_{1}g_{i1}<p_{1}'g_{i1}$ implies $q^{*}_{i1} > q'_{i1}$. This is true for all $i \in \mathcal{I}_{1}$ and it is also true that $q^{*}_{i2}> q_{i2}'$ for all $i \in \mathcal{I}_{2}$. This may not be true for the one user who has $q^{*}_{k1}>0, q^{*}_{k2}>0$ (such a user may or may not exist depending on the problem instantiation). For this user $p^{*}_{1}$ and $p^{*}_{2}$ do not uniquely define $q^{*}_{k1}$ and $q_{k2}^{*}$. However, by considering (\ref{eqn:derivZero}) we see that it cannot be that $q^{*}_{i1}<q_{i1}'$ and $q^{*}_{i2}<q_{i2}'$, again since the derivative of $v_{k}$ is a decreasing function of $x_{k}$. So, either $q^{*}_{i1}>q_{i1}'$ or $q^{*}_{i2}>q_{i2}'$. Suppose $q^{*}_{i1}>q_{i1}'$ which leads to $Q_{1}=\sum_{i=1}^{I}q^{*}_{i1}>\sum_{i=I}^{I}q'_{i1}$. Then, vector $\mathbf{q'}$ is not the maximizing vector since it violates equality (\ref{eqn:clearing}). Similar reasoning holds for $p^{*}_{j}>p_{j}'$ $j=1,2$. Cases  2 and 3 are dealt in a similar manner.

Case 2: $\frac{p^{*}_{1}}{p^{*}_{2}}>\frac{p_{1}'}{p_{2}'}$. Then, either $p^{*}_{1}>p_{1}'$ or $p_{2}<p_{2}'$ (or both). Assume $p^{*}_{1}>p_{1}'$. Due to the decrease in the price ratio, there will be two effects: some users who are in $\mathcal{I}_{2}$ may switch to $\mathcal{I}_{1}'$. Also, for all $i \in \mathcal{I}_{1}'$, $q_{i1}^{*}<q'_{i1}$. This implies, $Q_{1}=\sum_{i=1}^{I}q^{*}_{i1}<\sum_{i=I}^{I}q'_{i1}$ and hence $\mathbf{q}'$ cannot be the maximizing vector. Similar reasoning goes for $p^{*}_{2}<p_{2}'$. The argument does not depend on whether there is a user who has who has $q^{*}_{k1}>0, q^{*}_{k2}>0$.
Case 3 is analogous to Case 2. We conclude that there exists a unique maximizing solution to SYSTEM, and a corresponding unique pair $p^{*}_{1}, p^{*}_{2}$ of Lagrange multipliers (we write $\mathbf{p}^{*}=[p_{1}^{*}\;p_{2}^{*}]$). 

The conditions to be fulfilled by the maximizing solution $\mathbf{q}$ for some $p_{1},p_{2}$ can be summarized as follows:
\begin{align}
\frac{\partial v_{i}(x_{i})}{\partial x_{i}}\frac{1}{g_{ij}}-p_{j} \leq &0,\; j=1,2;\; i=1,\ldots,I \label{eqn:kkt1} \\
q_{ij}\left(\frac{\partial v_{i}(x_{i})}{\partial x_{i}}\frac{1}{g_{ij}}-p_{j} \right)=&0,\; j=1,2;\; i=1,\ldots,I \label{eqn:kkt2}\\
\frac{q_{i1}}{g_{i1}}+\frac{q_{i2}}{g_{i2}} = & x_{i},\; i=1,\ldots,I \label{eqn:kkt3}\\
\sum_{i=1}^{I}q_{ij}=&Q_{j},\; j=1,2 \label{eqn:kkt4}\\
p_{j}>0,\; q_{ij} \geq& 0\; j=1,2;\; i=1,\ldots,I. \label{eqn:kkt5}
\end{align}
Hence, any vector $\mathbf{q}$ that maximizes SYSTEM fulfills equations (\ref{eqn:kkt1})-(\ref{eqn:kkt5}) for some $\mathbf{p}=[p_{1}\;p_{2}]$; conversely any vector tuple $\mathbf{q},\mathbf{p}$ that fulfills equations (\ref{eqn:kkt1})-(\ref{eqn:kkt5}) has the property that $\mathbf{q}$ is the solution to SYSTEM. Above we proved that such a vector tuple $\mathbf{q},\mathbf{p}$ exists and is unique.

The proof does not end here as we now need to show that the Lagrange multipliers, when interpreted as prices, will lead to users choosing the correct values of $q_{ij}$. 

Assume that $p^{*}_{1},p^{*}_{2}>0$ are the Lagrange multipliers identified with the vector that maximizes SYSTEM. Assume that the two networks are then charging prices $p^{*}_{1},p^{*}_{2}$ to their users, such that each user is facing a problem USER$_{i}(v_{i},p^{*}_{1},p^{*}_{2})$:
\begin{align}
\max & \; v_{i}\left(\frac{q_{i1}}{g_{i1}}+\frac{q_{i2}}{g_{i2}}\right)-\sum_{j=1}^{2}p^{*}_{j}q_{ij}\\
\text{over }& q_{ij}\geq 0 \nonumber
\end{align}
Then the maximizing condition is the same as (\ref{eqn:derivZero}) and (\ref{eqn:derivLessZero}):
\begin{align*}
\frac{\partial v_{i}(\frac{q_{i1}}{g_{i1}}+\frac{q_{i2}}{g_{i2}})}{\partial q_{ij}}=\frac{\partial v_{i}(x_{i})}{\partial x_{i}}\frac{1}{g_{ij}}-p^{*}_{j} = &0,\;q_{ij }>0  \\
\leq & 0,\; q_{ij} =0. 
\end{align*}
In particular, for all users (except possibly one) the maximizing values $q^{'}_{i1},q^{'}_{i2}$ are \textit{uniquely} defined and identical to the ones maximizing SYSTEM (since they come from the same set of equations). However, if there is a user $k$ such that $\frac{g_{k1}}{g_{k2}}=\frac{p^{*}_{2}}{p^{*}_{1}}$ then there are many values of $q_{k1}, q_{k2}$  that fulfill 
\begin{align*}
\frac{\partial v_{i}(x)}{\partial x}\Big\vert_{x=\frac{q_{i1}}{g_{i1}}+\frac{q_{i2}}{g_{i2}}}\frac{1}{g_{ij}}-p^{*}_{j} = 0 , j=1,2
\end{align*}
so in general \textbf{it is not enough that users be given prices alone} in order to find the social utility maximizing values of $\mathbf{q}'$. What is required is that the networks impose the supply limiting constraint: $\sum_{i=1}^{I}q_{ij}=Q_{j}$. Then, the remaining two values can be \textit{uniquely} found as $q'_{kj}=Q_{j}-\sum_{i\neq k}^{¥}q'_{ij}$.

Then, the vector $\mathbf{q'}$ is \textbf{unique}, and satisfies the following conditions:
\begin{align}
\frac{\partial v_{i}(x_{i})}{\partial x_{i}}\frac{1}{g_{ij}}-p^{*}_{j} \leq &0,\; j=1,2;\; i=1,\ldots,k-1,k+1,\ldots,I \label{eqn:pricing1} \\
q_{ij}\left(\frac{\partial v_{i}(x_{i})}{\partial x_{i}}\frac{1}{g_{ij}}-p^{*}_{j} \right)=&0,\; j=1,2;\; i=1,\ldots,k-1,k+1,\ldots,I \label{eqn:pricing2} \\
\frac{q_{i1}}{g_{i1}}+\frac{q_{i2}}{g_{i2}} = & x_{i},\; i=1,\ldots,I \label{eqn:pricing3}\\
\sum_{i=1}^{I}q_{ij}=&Q_{j},\; j=1,2 \label{eqn:pricing4}\\
p^{*}_{j}>0,\; q_{ij} \geq& 0\; j=1,2;\; i=1,\ldots,I. \label{eqn:pricing5} 
\end{align}
Equations (\ref{eqn:pricing1})-(\ref{eqn:pricing5}) are almost identical to equations (\ref{eqn:kkt1})-(\ref{eqn:kkt5}), except that we cannot claim yet that (\ref{eqn:pricing1}) and (\ref{eqn:pricing2}) are valid for user $k$ (it will turn out that they are). It remains to be shown that $q_{k1},q_{k2}$ constructed from $q^{*}_{kj}=Q_{j}-\sum_{i\neq k}^{¥}q^{*}_{ij}$ are also fulfilling equations (\ref{eqn:kkt1}) and (\ref{eqn:kkt2}).

We know that  there is a unique vector tuple $\mathbf{q}^{*}, \mathbf{p^{*}}$ such that $\mathbf{q}^{*}$ maximizes SYSTEM. Then, $\mathbf{q}^{*}$ fulfills all the conditions (\ref{eqn:kkt1})-(\ref{eqn:kkt5}) for $p^{*}_{1},p^{*}_{2}$. At the same time, $\mathbf{q}^{*}$ satisfies conditions (\ref{eqn:pricing1})-(\ref{eqn:pricing5}) (since they are a subset of conditions (\ref{eqn:kkt1})-(\ref{eqn:kkt5}) for the maximizing $p^{*}_{1},p^{*}_{2}$). On the other hand, as we have already shown, there exists a unique vector $\mathbf{q}'$ that satisfies (\ref{eqn:pricing1})-(\ref{eqn:pricing5}). Hence $\mathbf{q}^{*}=\mathbf{q}'$ and we can conclude that the vector that satisfies (\ref{eqn:pricing1})-(\ref{eqn:pricing5}) also satisfies (\ref{eqn:kkt1})-(\ref{eqn:kkt5}) and hence solves SYSTEM. 
\end{proof}

The Theorem \ref{th:PoA} is a direct corollary of Theorem \ref{thm:system}, and Propositions \ref{pr:integerMCE} and \ref{pro:eps-equilib}. \\

There is a curious resemblance between our problem and the problems treated by Kelly in his seminal paper \cite{Kelly:1997} (hence some of the resemblance in the naming of the maximization problems). In fact, we first considered reducing the problem we are treating to the one that Kelly was solving in his paper. It turns out that it cannot be done. The reason is that in Kelly's paper a single network proposes prices to users, and then these users can find a unique maximizer to their own optimization problem. In our case, a user may exist that cannot provide unique values to maximize his personal optimization problem (i.e. the undecided user). Hence, we necessarily need network intervention in order for this user to split his demand in the right way. 

We can also view our model as a Stackelberg game or a two-stage extensive game. The providers announce the prices during the first stage, and the users determine their resource requests during the second stage. A user $k$'s strategy is contingency plan for all possible choices, and thus a user may have infinite optimal strategies because it has infinite optimal choices of resource requests whenever $p_1 g_{k1} = p_2 g_{k2}$. For a given fixed price pairs $(p_1,p_2)$, however, at most one user will have the chance to choose among its infinite optimal resource requests (there will be at most one undecided user); any other user will have a unique optimal resource request. What we show in Theorem \ref{th:PoA} is that there exists a unique subgame perfect Nash equilibrium where the optimal strategies of the users should lead to balance of supply and demand for each provider at the equilibrium prices.

\end{appendix}

	\bibliography{twoNetCompetitionLong}
	\bibliographystyle{IEEEtran}

\end{document}